\documentclass{article}

\usepackage[final]{hyll-bio} 

\usepackage{amssymb} 
\usepackage{amsthm}  
\usepackage[cmex10]{amsmath}
\usepackage{array}
\usepackage{url}
\usepackage{cite}
\newtheorem{theorem}{Theorem}

\newtheorem{proposition}[theorem]{Proposition}
\newtheorem{corollary}[theorem]{Corollary}

\newtheorem{definition}[theorem]{Definition}

\newtheorem*{hyllprop}{Proposition}

\begin{document}

\title{A Logical Framework for Systems Biology}

\author{
Elisabetta de Maria \\
Univ. of Nice - Sophia-Antipolis  \\ 
edemaria@i3s.unice.fr
\and 
Jo{\"e}lle Despeyroux \\
INRIA and CNRS \\
joelle.despeyroux@inria.fr
\and
Amy P.~Felty \\
University of Ottawa  \\ 
afelty@eecs.uottawa.ca
}

\date{{\em Research Report} --- Version of \today} 

\maketitle

\begin{abstract}
We propose a novel approach for the formal verification of biological systems 
based on the use of a modal linear logic.
We show how such a logic can be used, with worlds as instants of time,
as an unified framework to encode both biological systems 
and 
temporal properties of their dynamic behaviour.
To illustrate our methodology, 
we consider a model of the P53/Mdm2 DNA-damage repair mechanism. 
We prove several properties that are important for such a model to satisfy 
and serve to illustrate the promise of our approach.  
We formalize the proofs of these properties in the Coq Proof Assistant, 
with the help of a Lambda Prolog prover for partial automation of the proofs. 
\end{abstract}


\section{Introduction}
\label{sec:intro}


In this paper, we consider the question of reasoning about 
biological systems in a modal linear logic.  We show that a new logic,
called Hybrid Linear Logic (HyLL) developed by the second author in
joint work with 
K. Chaudhuri~\cite{ChaudhuriDespeyroux:13tr,ChaudhuriDespeyroux:14}, 
is particularly well-suited to this purpose.  HyLL provides a unified
framework to encode biological systems, to express temporal properties
of their dynamic behaviour, and to prove these properties.  
By constructing proofs in the HyLL logic, 
we directly witness reachability as logical entailment.  
%
This approach is in contrast to
most current approaches to applying formal methods to systems biology,
which generally 
encode biological systems either in a dedicated programming language or in 
differential equations,
express properties in a temporal logic, and then
verify these properties against some form of traces built using an
external simulator.  In the next subsection, we review in some detail
the state of the art of such approaches, in order to further situate
and motivate our new approach.  In subsection~\ref{subsec:ll}, we
motivate our choice of linear logic in general and HyLL in
particular.  Then in subsection~\ref{subsec:contributions}, we further
outline our contributions as well as the overall organization of the
rest of the paper.

\subsection{Formal Methods for Systems Biology}


Computational systems biology provides a variety of methods for understanding
the structure of biological systems and for studying their dynamics, 
that is, the temporal evolution of the involved entities. 
%
According to the chosen abstraction level, several formalisms have been proposed 
in the literature to model biological networks
(e.g., gene regulatory networks, metabolic networks, or signal transduction networks).

%
%
To capture the qualitative nature of dynamics, Thomas introduced a Boolean approach for 
regulatory networks (an entity is present or absent) \cite{Thomas73jtb} and  
subsequently generalized it to multivalued levels of concentration \cite{TTK95bmb}. 
Contrary to Petri nets, which are based on synchronous updating techniques \cite{RML93ismb},
Thomas' discrete models are asynchronous. 
Other purely qualitative approaches are 
$\pi$-calculus \cite{RSS01psb},
bio-ambients \cite{RPSCS04tcs}, and reaction rules \cite{DanosLaneve:04-tcs,CCDFS04tcs}.

To describe the dynamics from a quantitative point of view, ordinary or
stochastic differential equations are heavily used. 
More recent approaches include hybrid Petri nets \cite{HT98isb} and automata \cite{ABIKMPRS01hscc}, 
stochastic $\pi$-calculus \cite{PC05tcsb}, and
rule-based languages with continuous/stochastic dynamics such as Kappa \cite{DanosLaneve:04-tcs},
%
Biocham \cite{FS08sfm}, or BioNetGen \cite{BFGH04bi}. 

The Biochemical Abstract Machine Biocham \cite{FSC04jbpc} is a
framework that allows the description of a biochemical system in terms
of reaction rules and the interpretion of it at different levels of abstraction, 
by either an asynchronous Boolean transition system (Boolean semantics), 
a continuous time Markov chain (stochastic semantics), 
or a system of ordinary differential equations over molecular concentrations (differential
semantics).
%
In this paper, taking inspiration from Biocham reaction rules to model regulatory networks,
we formalize the Biocham language (in the boolean case) in logic.

One of the most common approaches to the formal verification of biological systems is 
model checking \cite{C99mit}. 
Model checking allows one to verify desirable properties of a system by an exhaustive enumeration 
of all the states reachable by the system. 
In order to apply such a technique, the biological system should be encoded as a finite transition 
system and relevant system properties should be specified using propositional temporal logic. 
Formally, a transition system over a set $AP$ of atomic propositions is a tuple $M=(Q,T,L)$, 
where $Q$ is a finite set of states, $T \subseteq Q \times Q$ is a total transition relation (that
is, for every state $q \in Q$ there is a state $q' \in Q$ such that
$T(q,q')$), and $L: Q \rightarrow 2^{AP}$ is a labeling function that maps every state
into the set of atomic propositions that hold at that state.

Temporal logics are formalisms for describing sequences of
transitions between states \cite{E95else}.  
The computation tree logic
CTL$^\ast$ allows one to describe properties of computation trees. Its formulas are obtained by
(repeatedly) applying Boolean connectives, \emph{path quantifiers},
and \emph{state quantifiers} to atomic formulas. The path quantifier
$\tA$ (resp., $\tE$) can be used to state that all paths
(resp., some path) starting from a given state have some property.
The state quantifiers are 
the next time operator $\tX$, 
which can be used to impose that a property holds at the next state of a path, 
the operator $\tF$ (sometimes in the future), 
that requires that a property holds at some state on the path, 
the operator $\tG$ (always in the future), 
that specifies that a property is true at every state on the path, 
and the until binary operator $\tU$, 
which holds if there is a state on the path
where the second of its argument properties holds and, at every
preceding state on the path, the first of its two argument
properties holds. 
The branching time logic CTL is a fragment of CTL$^\ast$
that allows quantification over the paths starting from a given
state. Unlike CTL$^\ast$, it constrains every state quantifier to be
immediately preceded by a path quantifier. \
The linear time logic LTL is another known fragment of CTL$^\ast$ where one may only
describe events along a single computation path. 
Its formulas are of the form $\tA \varphi$, where $\varphi$ does not contain path
quantifiers, but it allows the nesting of state quantifiers. 
The Probabilistic Computation Tree Logic PCTL quantifies the different
paths by replacing the $\tE$ and $\tA$ modalities of CTL by probabilities.

In Biocham \cite{FS08sfm},
CTL, LTL, and a fragment of PCTL with numerical constraints are used in the three
semantics of reaction models, respectively, in the boolean semantics, in the differential
semantics and in the stochastic semantics.

Given a transition system $M=(Q,T,L)$, a state $q \in
Q$, and a temporal logic formula $\varphi$ expressing some desirable
property of the system, the \emph{model checking problem} consists
of establishing whether $\varphi$ holds at $q$ or not, namely,
whether $M, q \models \varphi$. Another formulation of the model checking problem
consists of finding all the states $q \in Q$ such that $M, q \models \varphi$. 
Observe that the second formulation is more general than the first one.

There exist several tools for checking if a finite state system verifies a given 
CTL, LTL, or PCTL formula, e.g., NuSMV \cite{CCG99cav}, SPIN \cite{HG03addison}, 
and PRISM \cite{HKNP06tacas}.

In contrast to the above approaches, in our new technique
we encode both biological systems and temporal properties in HyLL,
and prove that the properties can be derived from the system.
We focus on Boolean systems and in this case 
a time unit corresponds to a transition in the system.
We believe that discrete modeling is crucial in systems
biology because it allows taking into account some phenomena
that have a very low chance of happening (and could thus be neglected by
differential approaches), but which may have a strong impact on system
behavior.

\subsection{Linear Logic}
\label{subsec:ll}

Linear Logic (LL)~\cite{girard87tcs} is particularly well suited for describing 
state transition systems. 
LL  has been successfully used to model such diverse systems as: 
planning~\cite{saranli07icra}, 
Petri nets, CCS, the $\pi$-calculus~\cite{cervesato03tr,miller92welp}, 
concurrent ML~\cite{cervesato03tr}, security protocols~\cite{bozzano02phd}, 
multi-set rewriting, 
graph traversal algorithms~\cite{simmons08icalp}, 
and games.

In the area of biology, for example, a rule of activation 
(e.g., a protein activates a gene or the transcription of another protein)
can be modeled by the following LL axiom:
$$\cn{active}(a,b) \eqdef \cn{pres}(a) \limp (\cn{pres}(a) \otimes \cn{pres}(b)).$$
The formula $\cn{active}(a,b)$ describes the fact that a state where $\cn{pres}(a)$ is true
can evolve into a state where both  $\cn{pres}(a)$ and $\cn{pres}(b)$ are true.

Propositions such as $\cn{pres}(a)$ are called {\it resources}, and a rule in the logic 
can be viewed as a rewrite rule from a set of resources into another set of 
resources,
where a set of resources describes a state of the system.
Thus, a particular state transition system can be modeled by a set of rules of 
the above shape. The rules of the logic then allow us to prove some 
desired properties of the system, such as, for example, the existence of a stable state.

However, linear implication is
timeless: there is no way to correlate two concurrent transitions.
If resources have lifetimes and state changes have temporal, probabilistic or
stochastic \emph{constraints}, then the logic will allow inferences that may not
be realizable in the system being modeled. 
This was the motivation of 
the development of HyLL, which was designed
to represent constrained transition systems.

\subsection{Contributions and Organization}
\label{subsec:contributions}

In this work, we present some first applications of HyLL to systems
biology.  We present HyLL in Section~\ref{sec:hyll} and the overall
approach to the application domain in Section~\ref{sec:approach}.  
We choose a simple yet representative biological example 
concerning
the DNA-damage repair mechanism based on proteins p53 and Mdm2, 
and present and prove
several properties of this system (Section~\ref{sec:example}).  We
fully formalize these proofs in a theorem prover we have implemented
in the Coq Proof Assistant~\cite{BertotCasteran:2004} and
$\lambda$Prolog~\cite{MillerNadathur:2012}
(Section~\ref{sec:formal-proofs}). This prover is designed to both
reason in HyLL and to formalize meta-theoretic properties about it.

We discuss the merits
and eventual drawbacks of this new approach compared to approaches using 
temporal logic and model checking.
To better illustrate the correspondence with such approaches, 
which all use temporal logic to reason about 
(simulations of models of) the biological systems described, 
we also
present in some detail the encoding of temporal logic operators in HyLL 
(Section~\ref{sec:related-works}).
%

We conclude and discuss future work in Section~\ref{sec:conclusion}.
%
%

In appendix \ref{appendix:example.proofs}, we give the sequent proofs  
of the properties of our biological system;
the formalization of these proofs is available in our electronic 
appendix: \url{www.eecs.uottawa.ca/~afelty/fmmb14/}. 



\section{A Hybrid Linear Logic}
\label{sec:hyll}

HyLL is a conservative extension of intuitionistic
first-order linear logic (LL)~\cite{girard87tcs} where the truth judgements
are parameterized on a \emph{constraint domain}.
Instead of the ordinary judgement ``$A$ is true'', for a proposition $A$,
judgements of HyLL are of the form 
``$A$ is true under constraint $w$'', abbreviated as $A ~@~ w$. 
A typical example of such a judgement is ``$A$ is true at time $t$'', 
or ``with probability $p$.''


\subsection{HyLL Syntax}

Like in the linear logic LL, propositions are
interpreted as \emph{resources} which may be composed into a \emph{state} using
the usual linear connectives, and the linear implication ($\limp$) denotes a
transition between states. The world label $w$ of a judgement $A ~@~ w$ represents a
constraint on states and state transitions; particular choices for the worlds
produce particular instances of HyLL. 
The common component in all the instances of HyLL is the proof theory, 
which is fixed once and for all. 
%
The minimal requirement on the kinds of constraints that HyLL can deal with is defined as follows:

\begin{definition} \label{def:constraint-domain}
  A \emph{constraint domain} $\cal W$ is a monoid structure $\langle W, ., \iota\rangle$. 
  The elements of $W$ are called \emph{worlds}, 
  and the partial order $\preceq\ : W \times W$---defined as $u \preceq w$ if there exists 
  $v \in W$ such that $u . v = w$---is the \emph{reachability relation} in $\cal W$.
\end{definition}

\noindent
The identity world $\iota$ is $\preceq$-initial and is intended to represent the
lack of any constraints. Thus, the ordinary first-order linear logic is embeddable into any 
instance of HyLL by setting all world labels to the identity.
A typical and simple example of constraint domain is 
$\mathcal{T} = \langle \Re^+, +, 0\rangle$, or $\langle \N, +, 0\rangle$, 
representing instants of time. 


Atomic propositions are written using lowercase ($a, b, ...$) applied to
a sequence of \emph{terms} ($s, t, \ldots$), which are drawn from an untyped
term language containing term variables ($x, y, \ldots$) and function symbols
($f, g, ...$) applied to a list of terms. Non-atomic propositions are
constructed from the connectives of first-order intuitionistic linear logic and
the two hybrid connectives \emph{satisfaction} ($\texttt{at}$), which states that a
proposition is true at a given world ($w, \iota, u.v, \ldots$), and
\emph{localization} ($\downarrow$), which binds a name for the (current) world the 
proposition is true at. 
The following grammar summarizes the syntax of HyLL terms and propositions.

\smallskip
\bgroup
\begin{tabular}{l@{\ }r@{\ }l}
  $t$ & $::=$ & 
       $c ~|~ x ~|~ f(\vec t)$ \\
  $A, B$ & $::=$ & 
       $p(\vec t) ~|~ A \otimes B ~|~ \mathbf{1} ~|~ A \rightarrow B ~|~ 
                A \mathbin{\&} B ~|~ \top ~|~
       A \oplus B ~|~  \mathbf{0} ~|~ ! A ~|~ $ \\
       &  & $\forall x.~ A ~|~ \exists x.~ A ~|~
       (A ~\at~ w) ~|~ \downarrow u.~ A ~|~ \forall u.~ A ~|~ \exists u.~ A$ \\
\end{tabular}
\egroup

\smallskip
\noindent
Note that in the propositions $\downarrow u. A$, $\forall u.~ A$ and $\exists u.~ A$,
world $u$ is bound in $A$.
World variables cannot be used in terms, and neither can term variables occur in 
worlds; this restriction is important for the modular design of HyLL because it 
keeps purely logical truth separate from constraint truth.  
We let $\alpha$ range over variables of either kind.
%
Note that 
the $\downarrow$ connective commutes with every propositional connective, including 
itself \cite{ChaudhuriDespeyroux:13tr}.
That is, $\downarrow u.~ (A * B)$ is equivalent to 
$(\downarrow u.~ A) * (\downarrow u.~ B)$ for all binary connectives $*$, and
$\downarrow u.~ * A$ is equivalent to 
$* (\downarrow u.~ A)$ for every unary connective $*$, assuming the
commutation will not cause an unsound capture of $u$. 
It is purely a matter of taste where to place the $\downarrow$, and repetitions are harmless.
%
Note that $\downarrow$ and $\at$ commute freely with all non-hybrid 
connectives~\cite{ChaudhuriDespeyroux:13tr}.


\subsection{Sequent Calculus for HyLL}
\label{sec:hyll.sq}

We present the syntax of hybrid logic in a sequent calculus style, 
using Martin-L{\"o}f's principle of separating judgements (here: $A @ w$) 
and logical connectives (here: $\otimes, \limp$, ..., $\at$, ...) \cite{martin-lof96}.  
We use sequents of
%
the form $\Gamma; \Delta \vdashseq C ~@~ w$ where $\Gamma$ and $\Delta$ are 
sets of judgements of the form $A ~@~ w$, with $\Delta$ being moreover a \emph{multiset}.
$\Gamma$ is called the \emph{unrestricted context}: 
its hypotheses can be consumed any number of times.  
$\Delta$ is a \emph{linear context}: every hypothesis in it must be consumed singly in the proof.  
Note that in a judgement $A ~@~ w$ (as in a proposition $A ~\at~ w$), $w$ can be 
any expression in $\cal W$, not only a variable.
%
The notation $ A [ \tau / \alpha ]$ stands for the replacement of all free occurrences
of the variable $\alpha$ in $A$ with the expression $\tau$, avoiding capture. 
The expressions in the rules are to be read up to alpha-conversion.

\begin{figure*}[tp]   
\hspace{-0.5cm}
\framebox{\begin{minipage}{\textwidth}  

\input hyll-seq
\end{minipage}}
\caption{Sequent calculus for HyLL}
\label{fig:sq-rules}
\end{figure*}

The full collection of inference rules are in Fig.~\ref{fig:sq-rules}. 
The rules for the linear connectives are borrowed from \cite{ChangChaudhuriPfenning:03tr}
where they are discussed at length, so we omit a more thorough discussion here.
The rules for the first-order quantifiers are completely standard.  
%
%
A brief discussion of the hybrid rules follows. 
To introduce the \emph{satisfaction} proposition $(A ~\at~ u)$ (at any world $w$) 
on the right, the proposition $A$ must be true in the world $u$. 
The proposition $(A ~\at~ u)$ itself is then true at any world, not just
in the world $u$. In other words, $(A ~\at~ u)$ carries with it the world at which
it is true. Therefore, suppose we know that $(A ~\at~ u)$ is true (at any world $v$);
then, we also know that $A ~@~ u$, and we can use this hypothesis (rule ``at L''). 
These two introduction rules (on the right and on the lelft) match up precisely to 
(de)construct the information in the $A ~@~ w$ judgement.
%
The other hybrid connective of \emph{localisation}, $\downarrow$, is intended to be
able to name the current world. That is, if $\downarrow u.~ A$ is true at world $w$,
then the variable $u$ stands for $w$ in the body $A$. This interpretation is
reflected in its right introduction rule $\downarrow R$. For left introduction, 
suppose we have a proof of $\downarrow u.~ A ~@~ v$ for some world $v$. 
Then, we also know, and thus can use $A [v / u] ~@~ v$.

There are only two structural rules: the init rule infers an atomic initial sequent, 
and the copy rule introduces a contracted copy of an unrestricted assumption 
into the linear context (reading from conclusion to premise).
%
%
Weakening and contraction are admissible rules.
Their proofs are straightforward, by induction on the structure of the given 
derivations~\cite{ChaudhuriDespeyroux:13tr}.
%


\begin{theorem}[structural properties] \mbox{}
  \label{hyll-structural}
  \begin{itemize}
  \item If $\Gamma ; \Delta \vdashseq C ~@~ w$, then 
        $\Gamma, \Gamma' ; \Delta \vdashseq C ~@~ w$.  (weakening)
  \item If $\Gamma, A ~@~ u, A ~@~ u ; \Delta \vdashseq C ~@~ w$, then 
        $\Gamma, A ~@~ u ; \Delta \vdashseq C ~@~ w$. (contraction)
  \end{itemize}
\end{theorem}
%



%
The most important structural properties are the admissibility of the identity and  
cut theorem.
Thanks to the cut-admissibility theorem, 
the proof of the consistency of the logic is reduced to
the observation that there is no cut-free derivation of 
$. ; . \vdashseq \mathop{0} ~@~ w$.

The identity theorem is the general case of the init rule.
Its proof is straightforward, by induction on the structure of $A$.

\begin{theorem}[identity]
  \label{thm:identity}
  $\Gamma, A ~@~ w \vdashseq A ~@~ w$.
\end{theorem}

\begin{theorem}[cut] 
  \label{thm:cut} 
  $
   \\
   1.~ {If}~ \Gamma ; \Delta \vdashseq A ~@~ u 
   ~{and}~ \Gamma ; \Delta', A ~@~ u \vdashseq C ~@~ w, 
   ~{then}~ \Gamma ; \Delta, \Delta' \vdashseq C ~@~ w \\
   2.~ {If}~ \Gamma ; . \vdashseq A ~@~ u 
   ~{and}~ \Gamma, A ~@~ u ; \Delta \vdashseq C ~@~ w, 
   ~{then}~ \Gamma ; \Delta \vdashseq C ~@~ w.
$
\end{theorem}

\begin{proof} 
 By lexicographic structural induction on the given derivations, with cuts of
 kind 2 additionally allowed to justify cuts of kind 1. 
 See \cite{ChaudhuriDespeyroux:13tr} for the details.
\end{proof}

We can use the admissible cut rules to show that the following rules are
invertible: $\mathbf{1} L$, $\mathbf{0} L$, $\otimes L$, $\oplus L$, $! L$, $\exists L$, 
$\top R$, $\limp R$, $\mathbin{\&} R$, and $\forall R$.
In addition, the four hybrid rules, $\at R$, $\at L$, $\downarrow R$, and $\downarrow L$ are invertible
\cite{ChaudhuriDespeyroux:13tr}.

\begin{corollary}[consistency]
  \label{thm:consistency}
  There is no proof of  $ . ; . \vdashseq \mathop{0} ~@~ w$.
\end{corollary}
 
\begin{proof}
  Suppose $ . ; . \vdashseq \mathop{0} ~@~ w$ is derivable. Then, by the cut-admissibility
  theorem (\ref{thm:cut}) on the sequent calculus, the sequent $ . ; . \vdashseq \mathop{0} ~@~ w$ 
  must have a cut-free proof. 
  However we can see by simple inspection on the inference rules that this cannot be the case,
  as this sequent cannot be the conclusion of any rule of inference in the calculus. 
  Therefore, $ . ; . \vdashseq \mathop{0} ~@~ w$ is not derivable.
\end{proof}

Note also that
HyLL is conservative with respect to intuitionistic linear logic: as long as
no hybrid connectives are used, the proofs in HyLL are identical to those in LL~\cite{ChaudhuriDespeyroux:13tr}.
%

An example of derived statements, true in every semantics for worlds, is the following:

\begin{proposition}[relocalisation]
\label{thm:relocalisation}
Any true judgement can be relocated at any time in the future:
$$\dfrac{\Gamma ; A_1 ~@~ w_1 \cdots  A_k ~@~ w_k \vdashseq B ~@~ v}
        {\Gamma ; A_1 ~@~ {u . w_1} \cdots A_k ~@~ {u . w_k} \vdashseq B ~@~ {u . v}}$$
\end{proposition}
This property is particularly well suited to applications in biology.
The interested reader can find proofs and further meta-theoretical theorems about HyLL in \cite{ChaudhuriDespeyroux:13tr}.

\subsection{Some Definitions for Biology}  

We can define modal connectives in HyLL as follows:
%
\bgroup 
\begin{definition}[modal connectives] \label{def:modal-connectives} \mbox{} \\
\vspace{-1ex}
$$\begin{tabular}{l@{\ }l@{\ }l@{\ }l@{\ }l@{\ }l}
     $\Box A$ & $\eqdef$ ${\downarrow} u.~ \forall w.~ (A ~\at~ u . w)$ & \quad
     $\Diamond A$ & $\eqdef$ ${\downarrow} u.~ \exists w.~ (A ~\at~ u . w)$ & & \\
     $\delay{v} A$ & $\eqdef$ ${\downarrow} u.~ (A ~\at~ u . v)$ & \quad
     $\mathop{\dag} A$ & $\eqdef$ $\forall u.~ (A ~\at~ u)$ & & \\
\end{tabular}$$
\end{definition}
\egroup

\noindent
The connective $\delta$ represents a form of delay. Note its derived right rule:
$$\dfrac{\Gamma \vdashseq A ~@ ~ w . v} 
        {\Gamma \vdashseq \delay{v} A ~@~ w} ~[\delta R]$$

The proposition $\delay{v} A$ thus stands for an \emph{intermediate state} in a
transition to $A$. Informally it can be thought to be ``$v$ before $A$''.
%
\ednote[JD]{
 {\bf WAS}  ; thus,
$\Box A = \forall v. \delay{v} A$ represents \emph{all} intermediate states in the path to $A$,
and $\Diamond A = \exists v. \delay{v} A$  represents \emph{some} such state.  
\\ {\bf IS} $\Box A$ [resp. $\Diamond A$] represents all [resp. some] state(s) satisfying $A$ and reachable from now. }
The modally unrestricted proposition $\mathop{\dag} A$ represents a resource that is 
consumable in any world; it is mainly used to make transition rules applicable at all worlds.

\paragraph{}
It is worth remarking that HyLL proof theory can be seen as at least as
powerful as S5 \cite{ChaudhuriDespeyroux:13tr}.
Obviously HyLL is more expressive as it allows direct manipulation of the worlds using
the hybrid connectives: for example, the $\delta$ connective is not definable in S5.

\paragraph{} 
Oscillation is one of the typical properties of interest in biological systems
(illustrated here by Property 1 in Sect~\ref{sec:property1}).
%
%
%
%
In our logic, 
we can define one oscillation between $A$ and $B$, with respective delays $u$ and $v$, as follows:
\begin{definition}[one oscillation]
\label{def:one-oscillation} 

$\textrm{oscillate}_1\,(A,B,u,v) 
 \eqdef A ~\mathbin{\&}~ \delay{u} (B ~\mathbin{\&}~ \delay{v} A ) 
    ~\mathbin{\&}~ (A \mathbin{\&} B \limp \mathop{0}).$
\end{definition}
%

Note that the above HyLL proposition closely corresponds to the temporal formula
$A \land \tE \tF (B \land \tE \tF A)$.  

\paragraph{}
Oscillation can be more generally defined by the following proposition in HyLL:
\begin{definition}[oscillation]  

$\textrm{oscillate}_h\,(A,B,u,v) ~ \eqdef
\mathop{\dag} [ (A \limp \delay{u} B) ~\mathbin{\&}~ (B \limp \delay{v} A) ]
~\mathbin{\&}~ (A \mathbin{\&} B \limp \mathop{0})$.
\end{definition}
%

However, since oscillation can be considered a meta-level property of the biological systems
modeled in HyLL, this property is perhaps more naturally defined
as follows: 
\begin{definition}[oscillation] 
\label{def:meta-oscillation}
   $\textrm{oscillate}\,(A,B,u,v) ~ \eqdef{}$
   for any $w$, $(A ~@~ w \vdashseq B ~@~ w.u)$,
   $(B ~@~ w.u \vdashseq A ~@~ w.u.v)$,
   and $(\vdashseq  A \mathbin{\&} B \limp \mathop{0} ~@~ w)$.
\end{definition}

\subsection{Temporal Constraints}
\label{sec:hyllt}

In this paper, we only consider the constraint domain  
$\mathcal{T} = \langle \N, +, 0\rangle$ representing (discrete) instants of time,
and we write $\textrm{HyLL}[\mathcal{T}]$ for this instantiation of HyLL.
Delay (Definition~\ref{def:modal-connectives}) in $\textrm{HyLL}[\mathcal{T}]$ 
represents intervals of time;
 $\delay{d} A$ means ``$A$ will become available after delay $d$''.
This domain is very permissive because addition is commutative, 
resulting in the equivalence of $\delay{u} \delay{v} A$ and $\delay{v} \delay{u} A$.
The ``forward-looking'' connectives $\tG$ and $\tF$ of 
temporal logic
are precisely $\Box$ and $\Diamond$ of defn. \ref{def:modal-connectives}. 
For further discussion on the 
comparison with temporal logic and model checking, see Sec~\ref{sec:related-works}.
For our temporal specifications the notion of addition on times is fundamental.

Considering other constraint domains such as 
$\mathcal{T}_r = \langle \Re^+, +, 0\rangle$, or
probabilistic constraints (as discussed in \cite{ChaudhuriDespeyroux:13tr}) is left to future work.

\section{Approach}
\label{sec:approach}

In this work we take into consideration Boolean models consisting of 
(i) a set of Boolean variables, 
(ii) a (partially defined) initial state denoting the presence/absence of (some) variables, 
and (iii) a set of rules of the form
$L_i \Rightarrow R_i$, where the left (resp. right) hand side of the rule 
$L_i$ (resp. $R_i$) 
is the conjunction of a set of predicates concerning the presence/absence of variables. 
For example, $x_1 \wedge x_2 \Rightarrow \neg x_3$ is a valid rule. 
This kind of rule can be used to describe state transitions involving control variables or abstract processes. 
In Biocham \cite{FSC04jbpc}, they are mostly used to represent biochemical reactions 
(observe that Biocham rules are more restrictive, they do not 
express the absence of a variable). 
In this paper, we take advantage of these rules to express 
influence rules (e.g. activations and inhibitions) in a biological system.  
%

Observe that, given a Boolean model of this kind, although we do not do it,
it is always possible to build a transition system where 
the set of states is the set of all tuples of Boolean values denoting the presence/absence of the different variables,
and a pair of states ($s$,$s'$) belongs to the relation if and only if
there exists a rule $i$ such that $s$ satisfies $L_i$, $s'$ satisfies $R_i$, 
and all the variables not involved in rule $i$ have the same value in $s$ and $s'$. 

If $L_i$ speaks about the presence of a variable $x$, nothing can be said about the presence/absence of $x$ in $s'$.
In other words, nothing can be inferred about the consumption of variables appearing in the left hand side of rules. 
If we want to force consumption/not consumption, we have to specify it explicitly. 
Furthermore, our rules are asynchronous: one rule can be fired at a time.
If several rules can be taken from a given state, one of them is non-deterministically chosen.
As in Biocham, we choose an asynchronous semantics in order to eliminate the risk of affecting
fundamental biological phenomena such as the masking
of a relation by another one and the consequent inhibition/activation of biological processes.

To verify whether a Boolean model satisfies a given temporal property, 
our approach consists of encoding both the model and the property in the HyLL logic 
and producing a proof.
Observe that we do not explicitly build the transition system, 
we just give a set of variables and rules. 
Proving temporal properties can result in building a sub-part of the system. 
There is an analogy with on-the-fly model checking \cite{CVWY92FMSD}, 
a technique that in many cases avoids the construction of the entire state space of the system
(because the property to test guides the construction of the system). 
However, on-the-fly model checkers mainly deal with LTL formulas.  

\section{Example}
\label{sec:example}

%
In this section we focus on the P53/Mdm2 DNA-damage repair mechanism.
P53 is a tumor suppressor protein that is activated in reply to
DNA damage. In normal conditions, the concentration of p53 in the
nucleus of a cell is weak: its level is controlled by another
protein, Mdm2. These two proteins present a loop of negative
regulation. In fact, P53 activates the transcription of Mdm2 while the
latter accelerates the degradation of the former.

DNA damage increases the degradation rate of Mdm2 so
that the control of this protein on P53 becomes weaker and the concentration of p53 can
increase. P53 can thus exercise its functions, either stopping the cell cycle to allow DNA
repair, or provoking apoptosis, if damage is too heavy.
%
This is possible if the control of Mdm2 on p53 weakens. As written above, in
most of cases this happens for phosphorylation of p53 and Mdm2 through ATM.

When Mdm2 loosens its influence on P53, it is possible
to observe some oscillations of P53 and Mdm2 concentrations. The
answer to a stronger damage is a bigger number of oscillations.
In the literature, several models have been proposed to model the oscillatory behaviour of 
proteins P53 and Mdm2 
(see those introduced by
Chickermane et al. \cite{CRSN07jads}, by Ciliberto et al. \cite{CNT05cc}, 
and by Geva-Zatorsky et al. \cite{ZRIMSDYLPLA06msb}).
%
In the context of cancer therapies, in \cite{DFRS11tcs} it is shown
how a coupled model of the P53/Mdm2 DNA-damage repair mechanism, 
of the cell cycle, of the circadian clock, and of an anticancer drug 
can be a valuable tool to investigate the drug influence on the cell cycle. 

The concentration of P53 in healthy cells is weak  
because this protein is responsible for the activation of
many mechanisms that could be dangerous for cells. In an indirect
way, it stops the DNA synthesis process, it activates the production
of proteins charged with DNA reparation, and it can lead to apoptosis (cell death).

\subsection{Definition}
\label{sec:example.defn}

%
In the following we propose a simple Boolean model considering the presence/absence of the 
three variables DNAdam, P53, and MdM2. 
The corresponding transition system woud thus consist of eight states, 
each one associated to a 3-tuple of Boolean values denoting the presence/absence 
of the three variables.
Initial states are the ones where P53 is absent and Mdm2 is present.

The behaviour of the biological system is specified by the six following rules:
\begin{center}
\begin{tabular}{ll}
1) $\mbox{Dnadam}\Rightarrow\neg\mbox{Mdm2}$~~~~ &
  4) $\mbox{Mdm2}\Rightarrow\neg\mbox{P53}$ \\
2) $\neg\mbox{Mdm2}\Rightarrow\mbox{P53}$ &
  5) $\mbox{P53}\Rightarrow_C\neg\mbox{Dnadam}$ \\
3) $\mbox{P53}\Rightarrow\mbox{Mdm2}$ &
  6) $\neg\mbox{Dnadam}\Rightarrow\mbox{Mdm2}$
\end{tabular}
\end{center}
In rule 5, we use $\Rightarrow_C$ to force consumption, 
i.e., if P53 is present, after firing the rule both P53 and Dnadam are absent.
Note that, if Dnadam is present in the initial state, this rule (that refers to damage repair) 
can be non-deterministically fired after one, a few, or several P53/MdM2 oscillations. 
This is consistent with experiments, where the number of oscillations preceding damage repair 
depends on the damage entity.
In the other rules we assume there is no comsumption.
%

\subsection{Specification in HyLL}
\label{sec:example.specif}

The biological system is modeled in HyLL by a set of axioms of two kinds.
First, each rule of the biological system is modeled by a formula in HyLL, as it would be in ordinary linear logic,
with the additional use of the delay operator $\delay{v}$, making precise the delay taken by the corresponding transition. 
The domain of world is $\mathcal{T} = \langle \N, +, 0 \rangle$ and we fix all time delays to $1$.
Then we add a set of axioms stating some well-definedness conditions and the initial state.
To encode the basic Boolean model, we use two predicates:
$\cn{pres}(a)$ (seen earlier) and $\cn{abs}(a)$ to indicate the
presence or absence of variable $a$.
The full model is given in Fig.~\ref{fig:hyllsys}.

\subsubsection{Activation/Inhibition Rules}
\label{sec:activation-inhibition-rules} 

For the sake of clarity, we first define generic activation/inhibition actions.
These actions can be defined in various ways.
A first attempt at describing an activation rule without consumption, for example,
might be the following:
%
$$w\_\cn{active}(a,b) \eqdef \cn{pres}(a) \limp \delay{1}~ (\cn{pres}(a) \otimes \cn{pres}(b)).$$
%
\ednote[]{Weak rules.  Useful in principle, 
\\ althought the axiom 
   $\cn{well\_defined}_1(V) \eqdef \forall a \in V.~ [\cn{pres}(a) \oplus \cn{abs}(a)]$ 
   makes them not needed.
\\ None of our proofs makes use of them.
\\ We have no example yet which would need them.
}
%
This rule does not make any assumption on the value of $b$:
this kind of rule corresponds to the case of partial knowledge of the current state of the system.
Note however 
that in the case where $b$ is present, the above rule does not modify the state.
In order to avoid uninteresting proofs, we might alternatively define our activation rule in a more precise way,
as follows:
$$s\_\cn{active}(a,b) \eqdef \cn{pres}(a) \otimes \cn{abs}(b) \limp \delay{1} (\cn{pres}(a) \otimes \cn{pres}(b)).$$

We call the first kind of rules {\em weak} rules, and the second one {\em strong} rules.

\ednote[JD]{Find a better name for  {\em useless} rules. {\em looping} or {\em loop} rules?}
For the sake of completeness, let us mention a third kind of rules,
that we might call {\em useless},
although they are indeed needed to formalize loops in a state
(see our Property 3 below):
%
$$u\_\cn{active}(a,b) \eqdef \cn{pres}(a) \otimes \cn{pres}(b) \limp \delay{1} (\cn{pres}(a) \otimes \cn{pres}(b)).$$
The {\em general} form of an activation rule, taking in account the various possible values 
of $b$, and our eventual missing knowledge of this, 
is then the following: 
$$\begin{array}{rcl}
  \cn{active}(a,b) & \eqdef &
  (\cn{pres}(a) \oplus (\cn{pres}(a) \otimes \cn{pres}(b))
     \oplus (\cn{pres}(a) \otimes \cn{abs}(b)))
  \\ && \limp \delay{1}~ (\cn{pres}(a) \otimes \cn{pres}(b)).
\end{array}$$
The properties stated in the present paper will all use the general form of the biological rules, except 
Property 4, 
which is only valid for strong rules.
%

%
There are further alternatives for the  activation/\-inhibition rules.
Let us consider the above strong variant $s\_\cn{active}(a,b)$.
We can define an \emph{activation with consumption} (of the product $a$) as follows:
     $$s\_\cn{active}_c(a,b) \eqdef 
      \cn{pres}(a) \otimes \cn{abs}(b)  \limp \delay{1} (\cn{abs}(a) \otimes \cn{pres}(b)),$$
while a \emph{strong activation} will have an inhibitor effect, in case of absence of $a$:
    $$s\_\cn{active}_s(a,b) \eqdef 
      \cn{abs}(a) \otimes \cn{pres}(b) \limp \delay{1} (\cn{abs}(a) \otimes \cn{abs}(b)).$$

Our example also uses the corresponding three kinds of rules for the
inhibition actions (see Fig.~\ref{fig:hyllsys}).
Of course, we could also define activation/inhibition rules
accounting for a lack of information concerning consumption.

\subsubsection{The System}  

Before giving the complete definition of our system, 
we need to additionally specify, in each rule,  that if a variable is not touched, then its value remains the same
in the next state.
This is the purpose of the $\cn{unchanged}$ predicate,
which in turn leads us to introduce a parameter $\cn{vars}$, 
specifying the set of variables of the biological system.

Note that the definition of the $\cn{unchanged}$ predicate
relies on the hypothesis (discussed earlier, in Sec~\ref{sec:approach}) 
that only one rule of the biological system can fire at a time.
Note also the use of the $!$ and $\mathbin{\&}$ operators in its definition: this is an intuitionistic predicate.
%

We define our biological system (with rules in their general form) by
the set of HyLL axioms given in Fig.~\ref{fig:hyllsys}.
%
\begin{figure*}[tp]   
\hspace{-0.5cm}
\framebox{\begin{minipage}{\textwidth}
\begin{itemize}
\item{\emph{Variables}}:
  \begin{quote}
     $\cn{unchanged}(x,w)\eqdef 
           ~!~ [ ( \cn{pres}(x) ~\at~ w \limp \cn{pres}(x) ~\at~ w.1 ) ~ \mathbin{\&} ~
           ( \cn{abs}(x) ~\at~ w \limp \cn{abs}(x) ~\at~ w.1 ) ]$. \\
     $\cn{unchanged}(V,w) \eqdef \otimes_{x \in V} \cn{unchanged}(x,w)$.
\end{quote}
\item {\emph{Activation}}: 
  \begin{quote}
     $\cn{active}(V,a,b) \eqdef 
      (\cn{pres}(a) \oplus (\cn{pres}(a) \otimes \cn{pres}(b)) \oplus (\cn{pres}(a) \otimes \cn{abs}(b)))
      \\ \vphantom{.} \qquad \qquad \qquad \quad  ~~
      \limp \delay{1}~ (\cn{pres}(a) \otimes \cn{pres}(b)) ~ \otimes \downarrow u.~ \cn{unchanged}(V \setminus \{a,b\},u))$. 
\end{quote}
\item {\emph{Activation with consumption}}: 
  \begin{quote}
     $\begin{array}{ll}
     \cn{active}_c(V,a,b) \eqdef 
      & (\cn{pres}(a) \oplus (\cn{pres}(a) \otimes \cn{pres}(b)) \oplus (\cn{pres}(a) \otimes \cn{abs}(b))) \\
      & \limp \delay{1}~ (\cn{abs}(a) \otimes \cn{pres}(b)) ~ \otimes \downarrow u.~ \cn{unchanged}(V \setminus \{a,b\},u)). 
       \end{array}$
  \end{quote}
\item {\emph{Strong activation}}: 
  \begin{quote}
     $\begin{array}{ll}
    \cn{active}_s(V,a,b) \eqdef 
       & (\cn{abs}(a) \oplus (\cn{abs}(a) \otimes \cn{pres}(b)) \oplus (\cn{abs}(a) \otimes \cn{abs}(b))) \\
       & \limp \delay{1}~ (\cn{abs}(a) \otimes \cn{abs}(b)) ~ \otimes \downarrow u.~ \cn{unchanged}(V \setminus \{a,b\},u)).
       \end{array}$
\end{quote}
\item {\emph{Inhibition}}: 
  \begin{quote}
     $\begin{array}{ll}
      \cn{inhib}(V,a,b) \eqdef 
       & (\cn{pres}(a) \oplus (\cn{pres}(a) \otimes \cn{pres}(b)) \oplus (\cn{pres}(a) \otimes \cn{abs}(b))) \\
       &\limp \delay{1}~ (\cn{pres}(a) \otimes \cn{abs}(b)) ~ \otimes \downarrow u.~ \cn{unchanged}(V \setminus \{a,b\},u)). 
       \end{array}$
  \end{quote}
\item {\emph{Inhibition with consumption}}: 
  \begin{quote}
     $\begin{array}{ll}
    \cn{inhib}_c(V,a,b) \eqdef 
     & (\cn{pres}(a) \oplus (\cn{pres}(a) \otimes \cn{pres}(b)) \oplus (\cn{pres}(a) \otimes \cn{abs}(b))) \\
     & \limp \delay{1}~ (\cn{abs}(a) \otimes \cn{abs}(b)) ~ \otimes \downarrow u.~ \cn{unchanged}(V \setminus \{a,b\},u)).
     \end{array}$
 \end{quote}
\item {\emph{Strong inhibition}}: 
  \begin{quote}
     $\begin{array}{ll}
     \cn{inhib}_s(V,a,b) \eqdef 
     & (\cn{abs}(a) \oplus (\cn{abs}(a) \otimes \cn{pres}(b)) \oplus (\cn{abs}(a) \otimes \cn{abs}(b))) \\
     & \limp \delay{1}~ (\cn{abs}(a) \otimes \cn{pres}(b)) ~ \otimes \downarrow u.~ \cn{unchanged}(V \setminus \{a,b\},u)). 
     \end{array}$
 \end{quote}
\item{\emph{Well definedness}}: 
  \begin{quote}
  $\cn{well\_defined}_0(V) \eqdef 
        \forall a \in V.~ [\cn{pres}(a) \otimes \cn{abs}(a) \limp \mathop{0}]$. 
\\ 
  $\cn{well\_defined}_1(V) \eqdef 
       \forall a \in V.~ [\cn{pres}(a) \oplus \cn{abs}(a)]$. \\
   $\cn{well\_defined}(V) \eqdef 
        \cn{well\_defined}_0(V), \cn{well\_defined}_1(V)$.
  \end{quote}
\item{\emph{The system}}:
  \begin{quote}
     $\begin{array}{l@{\qquad}l}
      \cn{vars} \eqdef \{\cn{p53}, \cn{Mdm2},\cn{DNAdam} \}. & \\
      \cn{rule}(1) \eqdef \cn{inhib}(\cn{vars},\cn{DNAdam},\cn{Mdm2}). &
      \cn{rule}(4) \eqdef \cn{inhib}(\cn{vars},\cn{Mdm2},\cn{p53}). \\
      \cn{rule}(2) \eqdef \cn{inhib_s}(\cn{vars},\cn{Mdm2},\cn{p53}). &
      \cn{rule}(5) \eqdef \cn{inhib_c}(\cn{vars},\cn{p53},\cn{DNAdam}). \\
      \cn{rule}(3) \eqdef \cn{active}(\cn{vars},\cn{p53},\cn{Mdm2}). &
      \cn{rule}(6) \eqdef \cn{{inhib_s}(\cn{vars},\cn{DNAdam},\cn{Mdm2})}.
      \end{array}$
  \end{quote}
 \begin{quote}
     $\cn{system}\eqdef \cn{vars}, \cn{rule}(1), \cn{rule}(2), 
     \cn{rule}(3), \cn{rule}(4), \cn{rule}(5), \cn{rule}(6),
     \cn{well\_defined}(\cn{vars}).
     $
  \end{quote}
\item{\emph{Initial state}}:
  \begin{quote}
     $\cn{initial\_state}\eqdef \cn{abs}(\cn{p53}) \otimes \cn{pres}(\cn{Mdm2}),
  \quad 
     \cn{initial\_state} ~\at~ 0.
     $
    \end{quote}

\end{itemize} 
\end{minipage}}
\caption{Representation of the System in HyLL}
\label{fig:hyllsys}
\end{figure*}

\subsection{Proofs}
\label{sec:example.proofs}

Although linear logic is well suited to describing transition systems, as we do here in the area of  biology, 
this logic can sometimes be too precise in its resource management for our needs.
To solve this constraint, we sometimes make precise that we do not care about the value of some variables.
We define a $\cn{dont\_care}$  predicate for this purpose:
   \begin{quote} 
   $\cn{dont\_care}(x) \eqdef \cn{pres}(x) \oplus \cn{abs}(x)$ \qquad
   $\cn{dont\_care}(V) \eqdef  \otimes_{x \in V} \cn{dont\_care}(x)$. 
   \end{quote}

This predicate is used in the statement of two of the four properties we present in this paper
(Properties 1 and 4).

Additionally,
in some proofs, when we do not know the value of some variables of the system, 
we sometimes need to perform a case analysis on their two possible values, 
using the $\cn{well\_defined}_1$ predicate introduced for this purpose.

Finally let us give two definitions in order to further shorten propositions and proofs 
($state_0$ is a state equivalent to the initial state):
   \begin{quote}
   $state_0 \eqdef \cn{abs}(\cn{p53}) \otimes \cn{pres}(\cn{Mdm2})$ \\
   $state_1 \eqdef \cn{pres}(\cn{p53}) \otimes \cn{abs}(\cn{Mdm2}).$ 
   \end{quote}

\subsubsection{Property 1}
\label{sec:property1}
As long as there is DNA damage, the above system can oscillate (with a
short period) from $state_0$ to $state_1$ and back again.  
We outline the proof informally below.
The sequent proofs can be found in Appendix \ref{appendix:example.proofs}; 
we refer the reader to the electronic appendix for their formalization.

From $state_0$ and $\cn{pres}(\cn{DNAdam})$ 
we get $\cn{abs}(\cn{p53})$, $\cn{abs}(\cn{Mdm2})$, and $\cn{pres}(\cn{DNAdam})$ by rule $1$.
Then $\cn{pres}(\cn{p53})$, $\cn{abs}(\cn{Mdm2})$, and $\cn{pres}(\cn{DNAdam})$ ($state_1$) by rule $2$.
Then $\cn{pres}(\cn{p53})$, $\cn{pres}(\cn{Mdm2})$, and $\cn{pres}(\cn{DNAdam})$ by rule $3$,
and finally $\cn{abs}(\cn{p53})$, $\cn{pres}(\cn{Mdm2})$, and $\cn{pres}(\cn{DNAdam})$ ($state_0$) by rule $4$.

We define (and prove) our property in the two possible ways discussed earlier 
(Sec~\ref{sec:example.specif}),
roughly corresponding to Definitions~\ref{def:one-oscillation} 
and~\ref{def:meta-oscillation}, respectively.  
The difference here is that our initial state (the one from which the oscillation starts)
includes the presence of DNA damage.

\begin{hyllprop}[Property 1, Version 1]
For any world $w$, 
there exists two worlds $u$ and $v$ such that both $u$ and $v$ are less than $3$ and 
the following holds: 
%

$ \begin{array}{l}
   \mathop{\dag} system ~@~ 0 ~;~ state_0 ~\otimes~ \cn{pres}(\cn{DNAdam})~@~ w \\
   \quad \vdashseq  
      \delay{u}~ 
      [( state_1 ~\otimes~ \cn{dont\_care}(\cn{DNAdam}) ) ~\mathbin{\&}~
      ( \delay{v}~ (state_0 ~\otimes~ \cn{dont\_care}(\cn{DNAdam})) )] ~@~ w 
   \end{array}
$
\end{hyllprop}


Aternatively, our property can be defined:
%
\begin{hyllprop}[Property 1, Version 2]
\label{prop:property1}
For any world $w$, 
there exists two worlds $u$ and $v$ such that both $u$ and $v$ are less than $3$ and the following holds: 
%

$
\begin{array}{l}
   \mathop{\dag} system ~@~ 0 ~;~ state_0 ~\otimes~ \cn{pres}(\cn{DNAdam}) ~@~ w 
   \vdashseq  state_1 ~\otimes~ \cn{dont\_care}(\cn{DNAdam}) ~@~ w.u 
~~\textrm{and}~~ \\
   \mathop{\dag} system ~@~ 0 ~;~ state_1 ~@~ w.u
   \vdashseq state_0 ~@~ w.u.v
\end{array}
$
\end{hyllprop}
%
There are no $\cn{dont\_care}$'s needed in the conclusion of the second sequent 
because only rules 3 and 4 are used, which don't involve DNAdam.

\subsubsection{Property 2} 
\label{sec:property2}
DNA damage can be quickly recovered. 

From $state_0$ and $\cn{pres}(\cn{DNAdam})$ 
we get $\cn{abs}(\cn{p53})$, $\cn{abs}(\cn{Mdm2})$, and $\cn{pres}(\cn{DNAdam})$ by rule $1$.
Then $\cn{pres}(\cn{p53})$ and $\cn{abs}(\cn{Mdm2})$ ($state_1$) and $\cn{pres}(\cn{DNAdam})$ by rule $2$.
Then $\cn{abs}(\cn{p53})$, $\cn{abs}(\cn{Mdm2})$, and $\cn{abs}(\cn{DNAdam})$ by rule $5$,
and finally $\cn{abs}(\cn{p53})$ and $\cn{pres}(\cn{Mdm2})$ ($state_0$) and $\cn{abs}(\cn{DNAdam})$ 
by rule $6$.

Our property can be stated directly as follows.

\begin{hyllprop}[Property 2]
\label{prop:prop2}
For any world $w$, 
there exists a world $u$ such that $u$ is less than $5$ and the following holds:

$\begin{array}{l}
\mathop{\dag} system ~@~ 0;~ 
    state_0 ~\otimes~ \cn{pres}(\cn{DNAdam}) ~@~ w
    \vdashseq state_0 \otimes \cn{abs}(\cn{DNAdam}) ~@~ w.u
\end{array}$
\end{hyllprop}

\subsubsection{Induction/Case Analysis}
\label{sec:induction}

Most of interesting proofs require case analysis or induction; this is the case for Properties
3 and 4 below.
More precisely, we need here \emph{case analysis on the set of fireable rules}.
We implement this by a case analysis on the interval $[1..6]$ of our six rules, 
together with a new predicate  $\cn{fireable}$ defining 
the necessary conditions for each rule to fire.
We shall also need the negation of this predicate: 
$\cn{not\_fireable}$\footnote{The definition of $\cn{not\_fireable}$ relies on the
$\cn{well\_defined}_1$ hypothesis  $\forall a.~ \cn{pres}(a) \oplus
\cn{abs}(a)$ for all the variables not involved in the hypothesis of the rules.}.
We give here the definitions of these predicates for the first rule of our system,
for both  (strong and general) styles of the rules used in the present paper
(the complete definitions can be found in Appendix \ref{appendix:example.specif}):
\begin{quote}
  $\begin{array}{ll}
   \cn{fireable_s}(1) & \eqdef \cn{pres}(\cn{DNAdam}) \otimes \cn{pres}(\cn{Mdm2}) 
       \otimes \cn{dont\_care}(\cn{p53}) \\
   \cn{not\_fireable_s}(1) & \eqdef 
        ( (\cn{abs}(\cn{DNAdam}) \otimes \cn{pres}(\cn{Mdm2})) 
          ~\oplus~  (\cn{pres}(\cn{DNAdam}) \otimes \cn{abs}(\cn{Mdm2})) \\
       & \quad ~ \oplus~  (\cn{abs}(\cn{DNAdam}) \otimes \cn{abs}(\cn{Mdm2})) )
        ~\otimes~ \cn{dont\_care}(\cn{p53}) \\\\
  \cn{fireable}(1) & \eqdef 
     ( \cn{pres}(\cn{DNAdam}) 
       ~\oplus~ (\cn{pres}(\cn{DNAdam})  \otimes \cn{pres}(\cn{Mdm2})) \\
     & \quad ~ \oplus~ (\cn{pres}(\cn{DNAdam})  \otimes \cn{abs}(\cn{Mdm2}))) 
       ~\otimes~ \cn{dont\_care}(\cn{p53}) \\
  \cn{not\_fireable}(1) & \eqdef \cn{abs}(\cn{DNAdam}) ~\otimes~ 
     \cn{dont\_care}(\{\cn{Mdm2},\cn{p53}\})
  \end{array}$
\end{quote}

An (informal) formula like 
``for any \textit{fireable rule} $r, P$''   
will be written as 
%
``for any rule $r$ in the interval $[1..6]$,  
  $(\cn{fireable}(r) ~\&~ P) \oplus \cn{not\_fireable}(r)$'',
which can be expressed fairly directly in the Coq formalization.
Note that both definitions of $\cn{fireable}$ and $\cn{not\_fireable}$
could be generated from the definitions of the biological rules, 
as the rules we consider in the present paper always have the same shape.

\subsubsection{Property 3} 
\label{sec:property3}
If there is no DNA damage, the system remains in the initial state.
A first attempt at formalizing this property might be:

For any world $w$, the following holds:
$\mathop{\dag} system ~@~ 0;~ \cn{abs}(\cn{DNAdam}) ~@~ 0 
 \vdashseq state_0 \otimes \cn{abs}(\cn{DNAdam}) ~@~ w . 
$

However, the above statement does not model our property.
We want to prove that 
if $\cn{abs}(\cn{DNAdam}) ~@~ 0$ 
then
$state_0 \otimes \cn{abs}(\cn{DNAdam}) ~@~ w$
holds, for all worlds $w$, 
\emph{no matter which rule is fired} to get to $w$. 
Thus our property requires a \emph{case analysis} on the rules of the biological system.

\begin{hyllprop}[Property 3]
\label{prop:property3}
Let $\mathcal{P}$ denote the formula $state_0  ~\otimes~ \cn{abs}(\cn{DNAdam})$.
For any world $w$, the following holds:
$\mathop{\dag} system ~@~ 0,~ \mathcal{P} ~@~ 0~  
  \vdashseq \mathcal{P} ~\at~ 0  ~@~ w;$

\noindent
and for any world $w$, for any rule $r$ in the interval $[1..6]$, 
the following holds: \\

$\begin{array}{l}
\mathop{\dag} system ~@~ 0 \vdashseq 
   \mathcal{P} \limp 
   (\cn{fireable}(r) \mathbin{\&} \delay{1} \mathcal{P})  ~\oplus~
\cn{not\_fireable}(r) ~@~ w
\end{array}$
\end{hyllprop}

The proof of the second statement proceeds by case analysis on the rules ($r$) of the biological system.
%

To prove that
  $\mathcal{P} ~@~ n \vdashseq  \mathcal{P} ~@~ n+1$, for each $\cn{fireable}$ rule $r$,
we have to prove that, if  
$\cn{abs}(\cn{p53}) \otimes \cn{pres}(\cn{Mdm2}) \otimes \cn{abs}(\cn{DNAdam})  ~@~ n$ 
then 
$\cn{abs}(\cn{p53}) \otimes \cn{pres}(\cn{Mdm2}) \otimes \cn{abs}(\cn{DNAdam}) ~@~ (n+1)$,
no matter which rule is fired.

At state $n$, $\cn{DNAdam}$ is absent, thus the only effet it can have is to 
strongly inhibate $\cn{Mdm2}$ ($\cn{rule}(6)$), making it present in the next state. 
However, $\cn{Mdm2}$ being already initialy present, this first applicable rule has no effect.

At state $n$, $\cn{p53}$ is also absent, thus no rule involving $\cn{p53}$ can apply.

At state $n$, the only present product is $\cn{Mdm2}$. The only (second) applicable rule is  
$\cn{inhib}(\cn{Mdm2},\cn{p53})$ ($\cn{rule}(4)$), which, after an interval of time $1$,
makes $\cn{p53}$ absent and leaves $\cn{Mdm2}$ present.
As $\cn{p53}$ was already absent, and $\cn{Mdm2}$ present, the next state is, here again,
equivalent to the previous state. 

Thus, there are only two fireable rules: $\cn{rule}(4)$ and $\cn{rule}(6)$,
and none of them modifies the current state of the system, which thus remains in its initial state.

\subsubsection{Property 4}
\label{sec:property4}
There is no path with two consecutive states where $\cn{p53}$ and $\cn{Mdm2}$ 
are both present or both absent.
In other words: 
from any state where $\cn{p53}$ and $\cn{Mdm2}$ are both present or both absent, 
we can only go to a state where either $\cn{p53}$ is present and $\cn{Mdm2}$ is absent or 
$\cn{p53}$ is absent and $\cn{Mdm2}$ is present.

This requires a stronger (natural) hypothesis: 
we need the property that each rule modifies at least one entity in the system.
In order to achieve this, we shall use the strong style of definitions for our inhibition and 
activation rules discussed earlier (sec.~\ref{sec:activation-inhibition-rules}).
For example, the activation rule will be defined as follows:
$$\begin{array}{l} 
s\_\cn{active}(V,a,b) \eqdef {} 
   \cn{pres}(a) \otimes \cn{abs}(b) 
   \limp {} 
   \delay{1} (\cn{pres}(a) \otimes \cn{pres}(b)) ~\otimes
      \downarrow u.~ \cn{unchanged}(V\setminus\{a,b\},u)).
\end{array}$$
The complete set of strong rules can be found in Appendix \ref{appendix:example.specif}.

Let $\mathcal{L}$ and $\mathcal{R}$ denote the following two formulas:
$$\begin{array}{l}
\mathcal{L} :=  (\cn{pres}(\cn{p53}) ~\otimes~ \cn{pres}(\cn{Mdm2})) ~\oplus~ {} 
   (\cn{abs}(\cn{p53}) ~ \otimes~ \cn{abs}(\cn{Mdm2}))\\
\mathcal{R} :=  ( (\cn{pres}(\cn{p53}) ~\otimes~ \cn{abs}(\cn{Mdm2})) ~\oplus~ {} 
   (\cn{abs}(\cn{p53}) ~\otimes~ \cn{pres}(\cn{Mdm2})) )
   \otimes {}  \cn{dont\_care}(\cn{DNAdam})
\end{array}$$
We want to prove that 
from state $\mathcal{L}$ we can only   
go to state $\mathcal{R}$, \emph{no matter which rule is fired}.
Here again, we need \emph{case analysis on the set of fireable rules}:
%
\begin{hyllprop}[Property 4]
\label{prop:property4}
For any world $w$, for any rule $r$ in the interval $[1..6]$, the following holds:
$$\begin{array}{l}
\mathop{\dag} system ~@~ 0 ;~ . \vdashseq \mathcal{L} \limp 
  (s\_\cn{fireable}(r) \mathbin{\&} \delay{1}
  \mathcal{R}) ~\oplus~ s\_\cn{not\_fireable}(r) ~@~ w
\end{array}$$
\end{hyllprop}

Property 4 could be written as the CTL formula
$\tA \tG(\mathcal{L} \rightarrow \tA \tX \mathcal{R})$ 
(see Sec \ref{sec:related-works} for the encoding of such a formula in HyLL).
Nevertheless, to simplify the proof we can observe that 
the argument property of $\tA \tG$ contains an implication and 
thus all the possible states verifying the left hand side of the implication 
should be taken into account in the proof.
As a matter of fact, at each step we do not make assumptions on the state where $\mathcal{L}$ holds,
we suppose to be in a generic state. 
The system satisfies Property 4 if all its states satisfy $\mathcal{L} \rightarrow \tA \tX \mathcal{R}$.
Since in our transition system all the states are reachable from the initial states, 
this corresponds to requiring the satisfaction of Property 4 at (that is, from) the initial states.
In case we want to test such a property at a state $S_i$ that is not connected to all the other states,  
we only need to prove the property in the subtree of the transition system rooted in $S_i$. 
Thus, we are only interested in the set of states reachable from $S_i$. 
In HyLL, we should prove the following theorem:
$if~ reachable(S,S_i)~ then~ S \vdash \mathcal{L} \rightarrow \forall r \in R~ \delta_1 \mathcal{R}$, 
  where $reachable(S,S_i)  \eqdef \exists u.~ S_i \rightarrow \delta_u S$.

\section{Formal Proofs}
\label{sec:formal-proofs}

As mentioned, we use a combination of the Coq Proof
Assistant and a theorem prover written in
the higher-order logic programming language
$\lambda$Prolog to formalize the properties
presented in the previous section.  We first describe the overall
approach (Sect.~\ref{subsec:combined}), and then describe the encoding
of the biological system in Coq (Sect.~\ref{subsec:biocoq}).

\subsection{A Combined Theorem Prover}
\label{subsec:combined}

Our approach is to fully formalize
proofs in Coq, using the $\lambda$Prolog prover to help with
partial automation of the proofs.  The $\lambda$Prolog prover is a
tactic-style interactive theorem prover implemented in the manner
described in~\cite{FeltyJAR:1993}, also given as an example in Sect.
9.4 of~\cite{MillerNadathur:2012}.  The code described there can be
viewed as a logical framework, implementing a tactic-style
architecture.  We use this code directly, and instantiate the
framework with tactics implementing the basic inference rules of HyLL.

We use Coq for two reasons.  First, we can build libraries in Coq that
allow us to reason at two levels.  We can prove meta-level properties of HyLL 
(for example, we have formalized
Theorem~\ref{hyll-structural} -weakening),   
and we can reason at the object-level,
which in this case means that we can prove HyLL sequents directly.  To
do so, we adopt the two-level style of reasoning used
in Hybrid~\cite{FeltyMomigliano:JAR10} where the logic we want to reason in
and about is called the \emph{specification logic} and is implemented
as an inductive predicate in Coq.  The inductive predicate in this
case defines HyLL sequents, and its definition is a fairly direct
modification of the ordered linear logic
in~\cite{FeltyMomigliano:JAR10}.  Second, once a proof is complete,
Coq provides a \emph{proof certificate}.  In particular, Coq
implements the calculus of inductive contstructions (CIC), where a
property is stated as a type in CIC and a proof is a $\lambda$-term
inhabiting that type.  This $\lambda$-term serves as a certificate,
which can be stored and checked independently.

It is, of course, possible to prove the properties in
Sect.~\ref{sec:example} only using Coq, but in general, proofs in Coq
of HyLL sequents quickly become cumbersome because of the amount of
detail required to apply each inference rule of HyLL.  The
$\lambda$Prolog prover is used to automatically infer much of this
detail.  For example, because we implement HyLL directly as an
inductive predicate in Coq, in order to apply the $\otimes L$ rule
(see Figure~\ref{fig:sq-rules}), Coq's \texttt{apply} tactic requires
arguments to be given explicitly for the instantiation of formulas $A$
and $B$, world $u$, and multisets $\Delta$ and $\Delta, A~@~u, B~@~u$.
As a result, the Coq proofs are verbose and often contain redundant
information.  In $\lambda$Prolog on the other hand, our primitive
tactics for applying HyLL inference rules use unification to infer
these arguments.  We could instead program a more automated version of
the \texttt{apply} tactic in Coq, tailored to applying HyLL rules
using Coq's Ltac facility, but this task would likely be more complex
and adhoc, since unification is not one of the primitive operations of
Coq.

It is also straightforward to represent HyLL proof terms in
$\lambda$Prolog, and implement the construction of these proof terms
as part of the implementation of the primitive tactics.  In our
$\lambda$Prolog prover, we construct such proof terms and then
translate them to strings representing Coq proof script.  This
translation is also implemented in $\lambda$Prolog.  These
automatically generated proof scripts are imported into Coq,
and then after some fine-tuning of the script, we obtain a proof
certificate for the entire proof.

To illustrate the approach, we briefly summarize our proof of
Proposition~\ref{prop:prop2} (Property 2) using this combined prover.
The statement of the property is expressed fairly directly in Coq, and
then after a few steps in Coq (very few in this case), the HyLL
sequent:
$$\mathop{\dag} system ~@~ 0;~ 
    state_0 ~\otimes~ \cn{pres}(\cn{DNAdam}) ~@~ w \\
\vdashseq state_0 \otimes \cn{abs}(\cn{DNAdam}) ~@~ w.4$$
is given to the $\lambda$Prolog prover.  (Currently this step is done
by hand, but we could of course automate the translation of a Coq
subgoal to a $\lambda$Prolog one.)  In the $\lambda$Prolog version $w$
is an arbitrary constant that does not appear anywhere else in the
sequent, generated by the meta-level universal quantifier of
$\lambda$Prolog.  After the proof is completed, the $\lambda$Prolog
prover generates and then outputs the Coq script to a file.  This file
is imported by hand into the Coq proof and modified slightly using a
small number of emacs macros as well as some fine tuning of proofs,
mainly about multiset equality.  Automating these last details is left
for future work.

Note that we do not use $\lambda$Prolog to fully automate proofs; the
$\lambda$Prolog prover is also interactive.  The user indicates what
HyLL rule to apply at each step, possibly with some other basic
information such as the position in $\Delta$ of the formula to which
the rule is to be applied.  The arguments needed to apply the rule are
then inferred automatically.  We could build more automation into the
prover, possibly in the style of the linear logic programming language
presented in~\cite{HodasMiller:I&C94} for general automation, with the
addition of heuristics tailored to our application.  
Again, this is left for future work.


\subsection{Representing the System in Coq}
\label{subsec:biocoq}

We summarize the encoding of our biological system in Coq.  This
section fills in some of the details of our approach to formalizing
proofs for readers familiar with interactive proof assistants and
inductive definitions.  Due to space restrictions, we do not describe
in detail the formalization of HyLL itself, but instead refer the
reader to~\cite{FeltyMomigliano:JAR10}, which uses
a \emph{higher-order abstract syntax} approach to representing
formulas and a specification logic for implementing inference rules,
as we do here.  We just note here that formulas of HyLL are encoded as
an inductive definition called \texttt{oo} and \texttt{seq} is the
inductive predicate defining the inference rules of HyLL.  Sequents
have the form \texttt{(seq Gamma Delta (A @ W))}.  In this
sequent, \texttt{Gamma} is a list of formulas of type \texttt{oo}
(using the built-in lists of Coq).  \texttt{Delta} is a multiset of
elements of type \texttt{oo}, where we build our own custom multiset
library.  Finally, \texttt{A} is a formula (type \texttt{oo})
and \texttt{W} is a world, where worlds are encoded using Coq's
built-in type for natural numbers.

In order to instantiate our encoding of HyLL to the example model we
consider here, we must introduce Coq sets representing the variables
and the biological system and the predicates expressing presence and
absence of these variables.  We define both as inductive types, and
also add some definitions for some convenient abbreviations.  Note
that constant names ending in an underscore in the Coq code below
represent direct instantiations (whose definitions we omit) of the
constants used to define the general version of HyLL.  For example,
the \texttt{atom\_} constant is the instantiated version of the part
of the HyLL encoding used to coerce atoms (type \texttt{atm}) to the
instantiated version of HyLL formulas (type \texttt{oo\_}).
\small
\begin{verbatim}
Inductive tm:Set := p53 | mdm2 | dNAdam.
Inductive atm:Set := pres_atm: tm -> atm | abs_atm: tm -> atm.
Definition vars := (p53::mdm2::dNAdam::nil).
Definition pres (x:tm): oo_ := (atom_ (pres_atm x)).
Definition abs (x:tm): oo_ := (atom_ (abs_atm x)).
\end{verbatim}
\normalsize
Definitions such as those in Sect.~\ref{sec:example.proofs} and
Figure~\ref{fig:hyllsys} are straightforward to encode in Coq.  As an
example, we show the definition of $\cn{dont\_care}$, both the direct
version and the version with multiple variable arguments.  The latter
is defined recursively using Coq's fixpoint operator.  The
symbols \texttt{+o} and \texttt{*o} represent the HyLL connectives
$\oplus$ and $\otimes$, respectively.
\small
\begin{verbatim}
Definition dont_care (t:tm): oo_ := (pres t +o abs t).
Fixpoint dont_cares (V:list tm): oo_ :=
  match V with
  | nil => One
  | (t::nil) => (dont_care t)
  | (t::ts) => ((dont_care t) *o (dont_cares ts))
  end.
\end{verbatim}
\normalsize
The definition of $\cn{unchanged}$ is similarly encoded, with
\texttt{unchanged1} and \texttt{unchanged} as the names of the direct
and multiple-argument versions, respectively.

We show one example each of an activation/inhibition predicate and
rule in Coq:
\small
\begin{verbatim}
Definition inhib (V:list tm) (a b:tm): oo_ :=
 ((pres a +o (pres a *o pres b) +o (pres a *o abs b)) ->>
  ((step (pres a *o abs b)) *o (Down (fun u => unchanged (minustm_ V (a::b::nil)) u)))).
Definition rule1 := inhib vars dNAdam mdm2.
\end{verbatim}
\normalsize
The symbols \texttt{->>}, \texttt{step}, and \texttt{Down} represent
the HyLL operators $\limp$, $\delay{1}$, and $\downarrow$,
respectively.  The $\downarrow$ operator binds a variable, and thus
its higher-order abstract syntax representation \texttt{Down} takes a
functional argument.  The function \texttt{minustm\_} implements
set difference for type \texttt{tm}.

In order to do induction and case analysis on rules, we define
$\cn{fireable}$ and $\cn{not\_fireable}$ as inductive predicates whose
first argument is the rule number.  The first clause of the general
versions of each is shown below.
\small
\begin{verbatim}
Inductive fireable: nat -> oo_ -> Prop :=
| f1: fireable 1 ((pres dNAdam +o (pres dNAdam *o pres mdm2) +o (pres dNAdam *o abs mdm2)) *o
                  (dont_care p53))
| ...

Inductive not_fireable: nat -> oo_ -> Prop :=
| nf1: not_fireable 1 (abs dNAdam *o (dont_cares (mdm2::p53::nil)))
| ...
\end{verbatim}
\normalsize
To illustrate their use, we state the Coq version of
\texttt{Gamma}
and \texttt{PP} are the Coq encodings of (resp.) $\mathop{\dag} system
~@~ 0$ and $(state_0  ~\otimes~ \cn{abs}(\cn{DNAdam}))$.
\small
\begin{verbatim}
Theorem Property3 : forall w:world, 
  seq Gamma ((PP @ 0)::nil) ((PP at 0) @ w) /\
  forall (n:nat) (A B:oo_), fireable n A -> not_fireable n B ->
  seq Gamma nil ((PP ->> ((A &a step PP) +o B)) @ w).
\end{verbatim}
\normalsize
The Coq proof of the second conjunct proceeds by case analysis
on \texttt{n}, followed by inversion on \texttt{(fireable n A)}
and \texttt{(not\_fireable n B)}, which provides instantiations
for \texttt{A} and \texttt{B}
(the conditions that express whether the rule is fireable or not).  
The resulting 6 subgoals are sent to
the $\lambda$Prolog prover, whose output is imported back into Coq as
described above.

\section{Comparison with model checking}
\label{sec:related-works}

%
While temporal logics such as LTL, CTL, or CTL$^\ast$ have been very successful in practice
with efficient model checking tools, the proof theory of these logics is very complex.
In contrast, HyLL has a very traditional proof theoretic pedigree:
it is presented in  the sequent calculus and enjoys cut-elimination and focusing~\cite{ChaudhuriDespeyroux:13tr}.
A further advantage of our approach with respect to model checking is that
it provides an unified framework to encode both transition rules  
and (both statements and proofs of)
temporal properties. 

Let us examine both approaches in more details.

\subsection{Temporal Operators}  

%
We propose the following encoding of temporal logic operators in $\textrm{HyLL}[\mathcal{T}]$, 
where $\mathcal{T} = \langle \N, +, 0\rangle$, representing instants of time.
While this domain does not have any branching structure like CTL, it is expressive enough 
for many common idioms because of the branching structure of derivations involving $\oplus$.

%
State quantifiers can be easily mapped.
There is a clear correspondence between $\tF$ (resp. $\tG$) and
$\Diamond$ (resp. $\Box$), see Definition \ref{def:modal-connectives}. 
The encodings of $\tX$ and $\tU$ are the following ones:
$\tX P \equiv \delta_1 P$ and $P_1 \tU P_2 \equiv 
\downarrow u.~ \exists v.~ P_2 ~\at~ u.v \otimes \forall w < v.~ P_1 ~\at~ u.w$.
where $P$, $P_1$, and $P_2$ are some propositions  (not necessarily
atomic ones).\footnote{Observe that a proposition characterizes a set of states.}   
%

As for path quantifiers, the question is more subtle.
The idea is that $\tE$ corresponds to the existence of a proof,
while for $\tA$ it is necessary to look at a proof considering all the possible rules 
to be applied at each step
(at each step of the proof, the chosen rule should not influence the property satisfaction).
%

We came to the conclusion that the encoding of $\tA$ in HyLL depends on the state quantifier following it.
Let R be the set of rules of our transition system. The mapping we propose is the following one:

\begin{itemize}
 \item $\tA \tX P$. In HyLL, we write $\forall r \in R~ \delta_1 P$.
More precisely, the encoding contemplating fireable rules is
$\forall r \in R~ (\cn{fireable}(r)~ \&~ \delta_1 P) \oplus \cn{not\_fireable}(r)$ 
(see Sect. \ref{sec:induction}).
For the sake of simplicity, in the following we omit such details concerning fireable rules.
\item $\tA \tG P$. 
It is equivalent to $P \wedge \tA \tG (P \rightarrow \tA \tX (P))$. In HyLL, 
we write $P \otimes \forall n (P~ at~ n) \rightarrow  \forall r \in R (P$ at $n+1)$.
\item $\tA \tF P$. 
It is equivalent to $P \vee \tA \tX (\tA \tF P)$. 
If we have a bound k on the number of steps needed to satisfy the property, 
we can expand this formula by obtaining: $P \vee \tA \tX (P \vee \tA \tX( \ldots \tA \tX P))$, 
with k nested occurrences of $\tA \tX$. 
In HyLL, we write 
$P \oplus \forall r \in R (\delta_1 P \oplus (\forall r \in R (\ldots \delta_{k} P)))$. 
Notice that another alternative is to express the $\tF$ operator by using 
$\tU$ ($\tF P \equiv \textrm{true}~ \tU P$).
\item $\tA(P_1 \tU P_2)$. 
It is equivalent to $P_2 \vee (P_1 \wedge \tA \tX (P_1 \tU P_2)$. 
If we have a bound k on the number of steps needed to satisfy the property, 
we can expand this formula by obtaining: 
$P_2 \vee (P_1 \wedge \tA \tX (P_2 \vee (P_1 \wedge \tA \tX (\ldots \tA \tX P_2))))$, 
with k nested occurrences of $\tA \tX$. 
In HyLL, we write 
$P_2 \oplus (P_1 \otimes \forall r \in R (\delta_1 P_2 
\oplus (\delta_1 P_1 \otimes \forall r \in R ( \ldots \delta_{k} P_2))))$.
\end{itemize}
%
%
In addition to the future connectives, 
the domain $\mathcal{T}$ also admits past connectives if we add
saturating subtraction ({\it i. e.}, $a - b = 0$ if $b \ge a$) to the language of worlds. 
We can then define the duals $\tH$ (historically) and $\tO$ (once) of $\tG$ and $\tF$ as:
$$\tH~P \eqdef  \downarrow u. \forall w. (P ~\at~ u - w) \qquad
    \tO~P \eqdef  \downarrow u. \exists w. (P ~\at~ u - w)
$$
%

\subsection{Model Checking}
%

%
A strength of our approach with respect to model checking
is that, when we prove an existential property using certain rules of a model,
we have the guarantee that \emph{all} the models containing such rules satisfy the
property. This is important because in biology we often deal with incomplete information.
It is also worth noting that in model checking, all objects are finite:
both the number of states, and the number of transitions in the state graph.
In HyLL, objects can potentially be infinite; in particular, we can have an infinite number of states.
%
Let us point out further advantages of our approach with respect to model checking. First of all,
when proving a given property  we do not need to blindly try all possible rules at each step
but we can guide the proof (see section \ref{sec:conclusion}).
Observe that a successful proof of a given property can be exploited to prove similar properties.
Furthermore, suppose we are able to prove a property of the system which is not desirable.
In this case the proof we get can help us in understanding what should be modified in the
system so that the property is not satisfied.
More precisely,
we can look for the rules to be removed/modified among those that have been used in the proof. 
%
%
In model checking, when a property turns out to be true, the reason is not investigated.
Finally, in~\cite{W83IC}, temporal logic is extended to allow the expression of
properties such as ``P is true at every even state of an infinite path.''  
A decision procedure for this extended logic is also defined,
but to the best of our knowledge, there is no model checking tool for it.  
In HyLL, if we add equality on worlds, we can write $\forall n=2k.~ P~ at~ n$.

Note that   
in some of the temporal properties we test, there is a bound on the number of time units,
and thus on the length of the proof. In this particular case,
there is a strong analogy with ``bounded model checking'' \cite{BCCSZ03academic},
where the user can limit the length of paths leading from the initial state to states 
which satisfy the property to be tested.
%
Observe that we could easily couple our model with other models sharing some variables.
As an example, we could consider a Boolean model of the cell cycle and add the following linking rules:
p53 activates p21 and p21 inhibits CycA and CycE, where the last three compounds belong to the cell cycle
model \cite{DFRS11tcs}.
%

A drawback of theorem proving with respect to model checking is that this method can be 
time consuming and needs an expert.
Recent advances in both proof theory and systems  
however provide us with at least partial, and sometimes complete, automation of the proofs.

\section{Conclusion and Future Work}
\label{sec:conclusion}


In this paper we argued that the HyLL logic can be successfully
exploited for formally verifying Boolean biological systems. 
This work is a first experiment along this new line of research 
(although we already provide fully mechanized proofs).
We focussed on a simple regulatory  network 
but our framework could be adopted to model several other kinds of
biological networks (e.g., neuronal, predator-prey, or ecological networks).

A natural extension of this work consists of applying our methodology to 
multivalued, continuous, and stochastic biological models. 
As far as the first case is concerned, the extension is straightforward, 
we just need to replace present/absent predicates by predicates indicating 
the discrete values of variables: $\cn{C}(A,x)$. 
The last two cases are more involved. 
We could try to use the domain of world 
${\mathcal W} =\langle \Re^+, +, 0\rangle$,  
use predicates $\cn{C}$ to represent variable concentrations 
and express the evolution of
each variable in terms of functions involving kinetic expressions 
such us the mass action law, Hill, or Michaelis-Menten kinetics. 
\cite{ChaudhuriDespeyroux:13tr} also discuss several alternatives 
for the worlds in the probabilistic case.
In any case, 
the challenge would consist of being able to perform symbolic calculus 
as much as possible without evaluating functions.

With regard to formal proofs, an alternative to using Coq with 
HyLL implemented as the specification logic 
is to use Abella~\cite{gacek09phd}, which also is a two-level system,
with $\mathcal{G}$ as the meta-logic in which one can define a
specification logic.  
This requires replacing Abella's current
higher-order intuitionistic specification logic with HyLL and
implementing support for building proofs using this logic.
%

Proofs of properties such as 1 and 2 require finding a path through
the system, which here means specifying a series of rules that can be
applied in a particular order.  At each step, there may be a choice
between several potential fireable rules.  In the interactive
proofs, such choices were made by hand. Our future work includes
building automated procedures (e.g., Coq tactics) to guide the proof.
%
%
%
We would also like to extend our model to include axioms for events
such as those considered in Biocham, which make it possible to
change the value of some variables under certain special conditions.
Such events often correspond to external inputs and have priority over the ordinary rules of a model.  
%

We were looking for the logical essence of biochemical reactions. 
What we envision for the domain of ``biological computation'' 
is a resource-aware stochastic or probabilistic $\lambda$-calculus 
that has HyLL propositions as (behavioral) types.
%

\section*{Acknowledgment}

This work was initially partially supported 
by the European TYPES project.
The second author thanks 
Fran\c{c}ois Fages, Sylvain Soliman, Alessandra Carbone, Vincent Danos, and 
Jean Krivine for fruitful discussions on various preliminary versions of 
the HyLL logic
in view of its potential applications to biology. 


\bibliographystyle{alpha}  
\bibliography{hyll-bio,betty,markov}


\clearpage
\appendix

\section{Example Specification in HyLL}
\label{appendix:example.specif}

%

\subsection{Strong Rules}

\begin{itemize}
\item{\emph{Variables}}:
  \begin{quote} 
     $\cn{unchanged}(x,w)\eqdef 
           ~!~ [ ( \cn{pres}(x) ~\at~ w \limp \cn{pres}(x) ~\at~ w.1 ) ~ \mathbin{\&} ~
           ( \cn{abs}(x) ~\at~ w \limp \cn{abs}(x) ~\at~ w.1 ) ]$. \\
      \hspace{-1.1cm}
     $\cn{unchanged}(V,w) \eqdef \otimes_{x \in V} \cn{unchanged}(x,w)$.
\end{quote}
\item {\emph{Activation}}: 
  \begin{quote}
    $s\_\cn{active}(V,a,b) \eqdef 
      \cn{pres}(a) \otimes \cn{abs}(b) 
      \limp \delay{1} (\cn{pres}(a) \otimes \cn{pres}(b)) \otimes \downarrow u.~ \cn{unchanged}(V \setminus \{a,b\},u))$. 
\end{quote}
\item {\emph{Activation with consumption}}: 
  \begin{quote}
     $s\_\cn{active}_c(V,a,b) \eqdef 
      \cn{pres}(a) \otimes \cn{abs}(b) 
      \limp \delay{1} (\cn{abs}(a) \otimes \cn{pres}(b))) \otimes \downarrow u.~ \cn{unchanged}(V \setminus \{a,b\},u))$. 
 \end{quote}
\item {\emph{Strong activation}}: 
  \begin{quote}
     $s\_\cn{active}_s(V,a,b) \eqdef 
      \cn{abs}(a) \otimes \cn{pres}(b) 
      \limp \delay{1} (\cn{abs}(a) \otimes \cn{abs}(b)) \otimes \downarrow u.~ \cn{unchanged}(V \setminus \{a,b\},u))$. 
 \end{quote}
\item {\emph{Inhibition}}: 
  \begin{quote}
     $s\_\cn{inhib}(V,a,b) \eqdef 
      \cn{pres}(a) \otimes \cn{pres}(b) 
      \limp \delay{1} (\cn{pres}(a) \otimes \cn{abs}(b)) \otimes \downarrow u.~ \cn{unchanged}(V \setminus \{a,b\},u))$. 
 \end{quote}
\item {\emph{Inhibition with consumption}}: 
  \begin{quote}
     $s\_\cn{inhib}_c(V,a,b) \eqdef 
      \cn{pres}(a) \otimes \cn{pres}(b) 
      \limp \delay{1} (\cn{abs}(a) \otimes \cn{abs}(b)) \otimes \downarrow u.~ \cn{unchanged}(V \setminus \{a,b\},u))$. 
 \end{quote}
\item {\emph{Strong inhibition}}: 
  \begin{quote}
     $s\_\cn{inhib}_s(V,a,b) \eqdef 
      \cn{abs}(a) \otimes \cn{abs}(b) 
      \limp \delay{1} (\cn{abs}(a) \otimes \cn{pres}(b)) \otimes \downarrow u.~ \cn{unchanged}(V \setminus \{a,b\},u))$. 
 \end{quote}
\item{\emph{Well definedness}}: 
  \begin{quote}
   $\cn{well\_defined}_0(V) \eqdef 
         \forall a \in V.~ [\cn{pres}(a) \otimes \cn{abs}(a) \limp \mathop{0}]$. \\
   $\cn{well\_defined}_1(V) \eqdef 
        \forall a \in V.~ [\cn{pres}(a) \oplus \cn{abs}(a)]$. \\
   $\cn{well\_defined}(V) ~ \eqdef 
        \cn{well\_defined}_0(V), \cn{well\_defined}_1(V)$.
  \end{quote}
\item{\emph{The system}} 
  \begin{quote}
  $\begin{array}{ll}
   \cn{vars}  & \eqdef \{\cn{p53}, \cn{Mdm2},\cn{DNAdam} \} \\
   s\_\cn{rule}(1) &  \eqdef s\_\cn{inhib}(\cn{vars},\cn{DNAdam},\cn{Mdm2}) \\
         &  \eqdef \cn{pres}(\cn{DNAdam}) \otimes \cn{pres}(\cn{Mdm2}) 
             \limp \delay{1} (\cn{pres}(\cn{DNAdam}) \otimes \cn{abs}(\cn{Mdm2})) \otimes \downarrow u.~ \cn{unchanged}(\cn{p53},u) \\
     s\_\cn{rule}(2) &  \eqdef \cn{Inhib_s}(\cn{vars},\cn{Mdm2},\cn{p53}) \\
           &  \eqdef \cn{abs}(\cn{Mdm2}) \otimes \cn{abs}(\cn{p53}) 
               \limp \delay{1} (\cn{abs}(\cn{Mdm2}) \otimes \cn{pres}(\cn{p53})) \otimes \downarrow u.~ \cn{unchanged}(\cn{DNAdam},u)) \\
     s\_\cn{rule}(3) &  \eqdef s\_\cn{active}(\cn{vars},\cn{p53},\cn{Mdm2}) \\
          &  \eqdef \cn{pres}(\cn{p53}) \otimes \cn{abs}(\cn{Mdm2}) 
              \limp \delay{1} (\cn{pres}(\cn{p53}) \otimes \cn{pres}(\cn{Mdm2})) \otimes \downarrow u.~ \cn{unchanged}(\cn{DNAdam},u)) \\
    s\_\cn{rule}(4) &  \eqdef s\_\cn{inhib}(\cn{vars},\cn{Mdm2},\cn{p53}) \\
           &  \eqdef \cn{pres}(\cn{Mdm2}) \otimes \cn{pres}(\cn{p53}) 
                \limp \delay{1} (\cn{pres}(\cn{Mdm2}) \otimes \cn{abs}(\cn{p53})) \otimes \downarrow u.~ \cn{unchanged}(\cn{DNAdam},u)) \\
     s\_\cn{rule}(5) &  \eqdef \cn{Inhib_c}(\cn{vars},\cn{p53},\cn{DNAdam}) \\
           &  \eqdef \cn{pres}(\cn{p53}) \otimes \cn{pres}(\cn{DNAdam}) 
               \limp \delay{1} (\cn{abs}(\cn{p53}) \otimes \cn{abs}(\cn{DNAdam})) \otimes \downarrow u.~ \cn{unchanged}(\cn{Mdm2},u)) \\
   \end{array}$
  $\begin{array}{ll}
    s\_\cn{rule}(6) &  \eqdef \cn{Inhib_s}(\cn{vars},\cn{DNAdam},\cn{Mdm2}) \\
          &  \eqdef \cn{abs}(\cn{DNAdam}) \otimes \cn{abs}(\cn{Mdm2}) 
              \limp \delay{1} (\cn{abs}(\cn{DNAdam}) \otimes \cn{pres}(\cn{Mdm2})) \otimes \downarrow u.~ \cn{unchanged}(\cn{p53},u)) 
   \end{array}$
\end{quote}
 \begin{quote}
     $\cn{system}\eqdef \cn{vars},s\_\cn{rule}(1),s\_\cn{rule}(2), 
    s\_\cn{rule}(3),s\_\cn{rule}(4),s\_\cn{rule}(5),s\_\cn{rule}(6),
     \cn{well\_defined}(\cn{vars}).
     $
  \end{quote}
\item{\emph{Initial state}}:
  \begin{quote}
     $\cn{initial\_state}\eqdef \cn{abs}(\cn{p53}) \otimes \cn{pres}(\cn{Mdm2}),\\
     \cn{initial\_state} ~\at~ 0.
     $
    \end{quote}
%
\item {\emph{Hypothesis (with strong rules)}}: 
  \begin{quote}
    $\cn{dont\_care}(x) ~ \eqdef \cn{pres}(x) \oplus \cn{abs}(x)$ \\
    $\cn{dont\_care}(V) \eqdef  \otimes_{x \in V} \cn{dont\_care}(x)$ 
 \end{quote}
  \begin{quote}
     $s\_\cn{fireable}(1) \eqdef \cn{pres}(\cn{DNAdam}) \otimes \cn{pres}(\cn{Mdm2}) 
        \otimes \cn{dont\_care}(\cn{p53}) \\
     s\_\cn{fireable}(2) \eqdef  \cn{abs}(\cn{Mdm2})  \otimes \cn{abs}(\cn{p53}) 
        \otimes \cn{dont\_care}(\cn{DNAdam}) \\
     s\_\cn{fireable}(3) \eqdef  \cn{pres}(\cn{p53}) \otimes \cn{abs}(\cn{Mdm2}) 
        \otimes \cn{dont\_care}(\cn{DNAdam}) \\
     s\_\cn{fireable}(4) \eqdef  \cn{pres}(\cn{Mdm2}) \otimes \cn{pres}(\cn{p53}) 
        \otimes \cn{dont\_care}(\cn{DNAdam}) \\
     s\_\cn{fireable}(5) \eqdef  \cn{pres}(\cn{p53}) \otimes \cn{pres}(\cn{DNAdam}) 
        \otimes \cn{dont\_care}(\cn{Mdm2}) \\
     s\_\cn{fireable}(6) \eqdef  \cn{abs}(\cn{DNAdam}) \otimes \cn{abs}(\cn{Mdm2}) 
        \otimes \cn{dont\_care}(\cn{p53}) \\
     $
  \end{quote}
     $s\_\cn{not\_fireable}(1) \eqdef \\ 
         ( (\cn{abs}(\cn{DNAdam}) \otimes \cn{pres}(\cn{Mdm2})) 
           ~\oplus~  (\cn{pres}(\cn{DNAdam}) \otimes \cn{abs}(\cn{Mdm2}))
           ~\oplus~  (\cn{abs}(\cn{DNAdam}) \otimes \cn{abs}(\cn{Mdm2})) )
         ~\otimes~ \cn{dont\_care}(\cn{p53}) \\
     s\_\cn{not\_fireable}(2) \eqdef \\
        ( (\cn{pres}(\cn{Mdm2})  \otimes \cn{abs}(\cn{p53})) 
          ~\oplus~ (\cn{abs}(\cn{Mdm2})  \otimes \cn{pres}(\cn{p53}))
          ~\oplus~ (\cn{pres}(\cn{Mdm2})  \otimes \cn{pres}(\cn{p53})) 
        ~\otimes~ \cn{dont\_care}(\cn{DNAdam}) \\
     s\_\cn{not\_fireable}(3) \eqdef \\
        ( (\cn{abs}(\cn{p53}) \otimes \cn{abs}(\cn{Mdm2}))
          ~\oplus~ (\cn{pres}(\cn{p53}) \otimes \cn{pres}(\cn{Mdm2}))
          ~\oplus~ (\cn{abs}(\cn{p53}) \otimes \cn{pres}(\cn{Mdm2})) )
        ~\otimes~ \cn{dont\_care}(\cn{DNAdam}) \\
     s\_\cn{not\_fireable}(4) \eqdef \\
        ( (\cn{abs}(\cn{Mdm2}) \otimes \cn{pres}(\cn{p53})) 
          ~\oplus~ (\cn{pres}(\cn{Mdm2}) \otimes \cn{abs}(\cn{p53})) 
          ~\oplus~ (\cn{abs}(\cn{Mdm2}) \otimes \cn{abs}(\cn{p53})) )
        ~\otimes~ \cn{dont\_care}(\cn{DNAdam}) \\
     s\_\cn{not\_fireable}(5) \eqdef \\
        ( (\cn{abs}(\cn{p53}) \otimes \cn{pres}(\cn{DNAdam}))
          ~\oplus~ (\cn{pres}(\cn{p53}) \otimes \cn{abs}(\cn{DNAdam}))
          ~\oplus~ (\cn{abs}(\cn{p53}) \otimes \cn{abs}(\cn{DNAdam})) )
        ~\otimes~ \cn{dont\_care}(\cn{Mdm2}) \\
     s\_\cn{not\_fireable}(6) \eqdef \\
        ( (\cn{pres}(\cn{DNAdam}) \otimes \cn{abs}(\cn{Mdm2}))
          ~\oplus~  (\cn{abs}(\cn{DNAdam}) \otimes \cn{pres}(\cn{Mdm2}))
          ~\oplus~  (\cn{pres}(\cn{DNAdam}) \otimes \cn{pres}(\cn{Mdm2})) )
        ~\otimes~ \cn{dont\_care}(\cn{p53}) 
     $
\end{itemize}

%

\newpage

\subsection{General Rules}

\begin{itemize}
\item{\emph{Variables}}:
  \begin{quote} 
     $\cn{unchanged}(x,w)\eqdef 
           ~!~ [ ( \cn{pres}(x) ~\at~ w \limp \cn{pres}(x) ~\at~ w.1 ) ~ \mathbin{\&} ~
           ( \cn{abs}(x) ~\at~ w \limp \cn{abs}(x) ~\at~ w.1 ) ]$. \\
      \hspace{-1.1cm}
     $\cn{unchanged}(V,w) \eqdef \otimes_{x \in V} \cn{unchanged}(x,w)$.
\end{quote}
\item {\emph{Activation}}: 
  \begin{quote}
     $\cn{active}(V,a,b) \eqdef 
      (\cn{pres}(a) \oplus (\cn{pres}(a) \otimes \cn{pres}(b)) \oplus (\cn{pres}(a) \otimes \cn{abs}(b)))
      \\ \vphantom{.} \qquad \qquad \qquad \quad  ~~
      \limp \delay{1}~ (\cn{pres}(a) \otimes \cn{pres}(b)) ~ \otimes \downarrow u.~ \cn{unchanged}(V \setminus \{a,b\},u))$. 
\end{quote}
\item {\emph{Activation with consumption}}: 
  \begin{quote}
     $\begin{array}{ll}
     \cn{active}_c(V,a,b) \eqdef 
      & (\cn{pres}(a) \oplus (\cn{pres}(a) \otimes \cn{pres}(b)) \oplus (\cn{pres}(a) \otimes \cn{abs}(b))) \\
      & \limp \delay{1}~ (\cn{abs}(a) \otimes \cn{pres}(b)) ~ \otimes \downarrow u.~ \cn{unchanged}(V \setminus \{a,b\},u)). 
       \end{array}$
  \end{quote}
\item {\emph{Strong activation}}: 
  \begin{quote}
     $\begin{array}{ll}
    \cn{active}_s(V,a,b) \eqdef 
       & (\cn{abs}(a) \oplus (\cn{abs}(a) \otimes \cn{pres}(b)) \oplus (\cn{abs}(a) \otimes \cn{abs}(b))) \\
       & \limp \delay{1}~ (\cn{abs}(a) \otimes \cn{abs}(b)) ~ \otimes \downarrow u.~ \cn{unchanged}(V \setminus \{a,b\},u)).
       \end{array}$
\end{quote}
\item {\emph{Inhibition}}: 
  \begin{quote}
     $\begin{array}{ll}
      \cn{inhib}(V,a,b) \eqdef 
       & (\cn{pres}(a) \oplus (\cn{pres}(a) \otimes \cn{pres}(b)) \oplus (\cn{pres}(a) \otimes \cn{abs}(b))) \\
       &\limp \delay{1}~ (\cn{pres}(a) \otimes \cn{abs}(b)) ~ \otimes \downarrow u.~ \cn{unchanged}(V \setminus \{a,b\},u)). 
       \end{array}$
  \end{quote}
\item {\emph{Inhibition with consumption}}: 
  \begin{quote}
     $\begin{array}{ll}
    \cn{inhib}_c(V,a,b) \eqdef 
     & (\cn{pres}(a) \oplus (\cn{pres}(a) \otimes \cn{pres}(b)) \oplus (\cn{pres}(a) \otimes \cn{abs}(b))) \\
     & \limp \delay{1}~ (\cn{abs}(a) \otimes \cn{abs}(b)) ~ \otimes \downarrow u.~ \cn{unchanged}(V \setminus \{a,b\},u)).
     \end{array}$
 \end{quote}
\item {\emph{Strong inhibition}}: 
  \begin{quote}
     $\begin{array}{ll}
     \cn{inhib}_s(V,a,b) \eqdef 
     & (\cn{abs}(a) \oplus (\cn{abs}(a) \otimes \cn{pres}(b)) \oplus (\cn{abs}(a) \otimes \cn{abs}(b))) \\
     & \limp \delay{1}~ (\cn{abs}(a) \otimes \cn{pres}(b)) ~ \otimes \downarrow u.~ \cn{unchanged}(V \setminus \{a,b\},u)). 
     \end{array}$
 \end{quote}
\item{\emph{Well definedness}}: 
  \begin{quote}
    $\cn{well\_defined}_0(V) \eqdef 
         \forall a \in V.~ [\cn{pres}(a) \otimes \cn{abs}(a) \limp \mathop{0}]$. \\
   $\cn{well\_defined}_1(V) \eqdef 
        \forall a \in V.~ [\cn{pres}(a) \oplus \cn{abs}(a)]$. \\
   $\cn{well\_defined}(V) ~ \eqdef 
        \cn{well\_defined}_0(V), \cn{well\_defined}_1(V)$.
  \end{quote}
\item{\emph{The system}} 
  \begin{quote}
   $\begin{array}{ll}
   \cn{vars} & \eqdef \{\cn{p53}, \cn{Mdm2},\cn{DNAdam} \} \\
   \cn{rule}(1) & \eqdef \cn{inhib}(\cn{vars},\cn{DNAdam},\cn{Mdm2}) \\
         & \eqdef (\cn{pres}(\cn{DNAdam}) 
             \oplus (\cn{pres}(\cn{DNAdam})  \otimes \cn{pres}(\cn{Mdm2})) 
             \oplus (\cn{pres}(\cn{DNAdam})  \otimes \cn{abs}(\cn{Mdm2}))) \\
          & \quad ~ \limp \delay{1} (\cn{pres}(\cn{DNAdam}) \otimes \cn{abs}(\cn{Mdm2})) ~ \otimes \downarrow u.~ 
             \cn{unchanged}(\cn{p53},u) \\
      \cn{rule}(2) & \eqdef \cn{inhib_s}(\cn{vars},\cn{Mdm2},\cn{p53}) \\
         & \eqdef (\cn{abs}(\cn{Mdm2}) 
             \oplus (\cn{abs}(\cn{Mdm2}) \otimes \cn{pres}(\cn{p53})) 
             \oplus (\cn{abs}(\cn{Mdm2}) \otimes \cn{abs}(\cn{p53}))) \\
         & \quad ~   \limp \delay{1} (\cn{abs}(\cn{Mdm2}) \otimes \cn{pres}(\cn{p53})) ~ \otimes \downarrow u.~ 
             \cn{unchanged}(\cn{DNAdam},u)) \\
  \end{array}$

  $\begin{array}{ll}
      \cn{rule}(3) & \eqdef \cn{active}(\cn{vars},\cn{p53},\cn{Mdm2}) \\
        & \eqdef (\cn{pres}(\cn{p53}) 
            \oplus (\cn{pres}(\cn{p53}) \otimes \cn{pres}(\cn{Mdm2})) 
            \oplus (\cn{pres}(\cn{p53}) \otimes \cn{abs}(\cn{Mdm2}))) \\
        & \quad ~   \limp \delay{1} (\cn{pres}(\cn{p53}) \otimes \cn{pres}(\cn{Mdm2})) ~ \otimes \downarrow u.~ 
           \cn{unchanged}(\cn{DNAdam},u)) \\
      \cn{rule}(4) & \eqdef \cn{inhib}(\cn{vars},\cn{Mdm2},\cn{p53}) \\
         & \eqdef (\cn{pres}(\cn{Mdm2})  
            \oplus (\cn{pres}(\cn{Mdm2}) \otimes \cn{pres}(\cn{p53})) 
            \oplus (\cn{pres}(\cn{Mdm2}) \otimes \cn{abs}(\cn{p53}))) \\
         & \quad ~   \limp \delay{1} (\cn{pres}(\cn{Mdm2}) \otimes \cn{abs}(\cn{p53})) ~ \otimes \downarrow u.~ 
             \cn{unchanged}(\cn{DNAdam},u)) \\
      \cn{rule}(5) & \eqdef \cn{inhib_c}(\cn{vars},\cn{p53},\cn{DNAdam}) \\
        & \eqdef (\cn{pres}(\cn{p53})  
            \oplus (\cn{pres}(\cn{p53}) \otimes \cn{pres}(\cn{DNAdam})) 
            \oplus (\cn{pres}(\cn{p53}) \otimes \cn{abs}(\cn{DNAdam}))) \\
         & \quad ~   \limp \delay{1} (\cn{abs}(\cn{p53}) \otimes \cn{abs}(\cn{DNAdam})) ~ \otimes \downarrow u.~ 
             \cn{unchanged}(\cn{Mdm2},u)) \\
      \cn{rule}(6) & \eqdef \cn{inhib_s}(\cn{vars},\cn{DNAdam},\cn{Mdm2}) \\
         & \eqdef (\cn{abs}(\cn{DNAdam})  
             \oplus (\cn{abs}(\cn{DNAdam}) \otimes \cn{pres}(\cn{Mdm2})) 
             \oplus (\cn{abs}(\cn{DNAdam}) \otimes \cn{abs}(\cn{Mdm2}))) \\
          & \quad ~   \limp \delay{1} (\cn{abs}(\cn{DNAdam}) \otimes \cn{pres}(\cn{Mdm2})) ~ \otimes \downarrow u.~ 
             \cn{unchanged}(\cn{p53},u))
   \end{array}$
 \end{quote}
\begin{quote} 
     $\cn{system}\eqdef \cn{vars}, \cn{rule}(1), \cn{rule}(2), 
     \cn{rule}(3), \cn{rule}(4), \cn{rule}(5), \cn{rule}(6),
     \cn{well\_defined}(\cn{vars}).
     $
  \end{quote}
\item{\emph{Initial state}}:
  \begin{quote}
     $\cn{initial\_state}\eqdef \cn{abs}(\cn{p53}) \otimes \cn{pres}(\cn{Mdm2}),\\
     \cn{initial\_state} ~\at~ 0.
     $
    \end{quote}
\item {\emph{Hypothesis}}: 
 \begin{quote}
     $\cn{dont\_care}(x) ~ \eqdef \cn{pres}(x) \oplus \cn{abs}(x)$ \\
     $\cn{dont\_care}(V) \eqdef  \otimes_{x \in V} \cn{dont\_care}(x)$ 
 \end{quote}
     $\cn{fireable}(1) \eqdef ( \cn{pres}(\cn{DNAdam}) 
             \oplus (\cn{pres}(\cn{DNAdam})  \otimes \cn{pres}(\cn{Mdm2}))
             \oplus (\cn{pres}(\cn{DNAdam})  \otimes \cn{abs}(\cn{Mdm2}))) 
             ~\otimes~ \cn{dont\_care}(\cn{p53}) \\
     \cn{fireable}(2) \eqdef  ( \cn{abs}(\cn{Mdm2}) 
             \oplus (\cn{abs}(\cn{Mdm2}) \otimes \cn{pres}(\cn{p53})) 
             \oplus (\cn{abs}(\cn{Mdm2}) \otimes \cn{abs}(\cn{p53})) )
             ~\otimes~ \cn{dont\_care}(\cn{DNAdam}) \\
     \cn{fireable}(3) \eqdef  ( \cn{pres}(\cn{p53}) 
            \oplus (\cn{pres}(\cn{p53}) \otimes \cn{pres}(\cn{Mdm2})) 
            \oplus (\cn{pres}(\cn{p53}) \otimes \cn{abs}(\cn{Mdm2})) ) 
             ~\otimes~ \cn{dont\_care}(\cn{DNAdam}) \\
     \cn{fireable}(4) \eqdef  ( \cn{pres}(\cn{Mdm2})  
            \oplus (\cn{pres}(\cn{Mdm2}) \otimes \cn{pres}(\cn{p53})) 
            \oplus (\cn{pres}(\cn{Mdm2}) \otimes \cn{abs}(\cn{p53})) )
             ~\otimes~ \cn{dont\_care}(\cn{DNAdam}) \\
     \cn{fireable}(5) \eqdef ( \cn{pres}(\cn{p53})  
            \oplus (\cn{pres}(\cn{p53}) \otimes \cn{pres}(\cn{DNAdam})) 
            \oplus (\cn{pres}(\cn{p53}) \otimes \cn{abs}(\cn{DNAdam})) )
             ~\otimes~ \cn{dont\_care}(\cn{Mdm2}) \\
     \cn{fireable}(6) \eqdef ( \cn{abs}(\cn{DNAdam})  
             \oplus (\cn{abs}(\cn{DNAdam}) \otimes \cn{pres}(\cn{Mdm2})) 
             \oplus (\cn{abs}(\cn{DNAdam}) \otimes \cn{abs}(\cn{Mdm2})) )
             ~\otimes~ \cn{dont\_care}(\cn{p53}) 
     $
\begin{quote}
  $\cn{not\_fireable}(1) \eqdef \cn{abs}(\cn{DNAdam}) ~\otimes~ 
      \cn{dont\_care}(\{\cn{Mdm2},\cn{p53}\}) \\
   \cn{not\_fireable}(2) \eqdef \cn{pres}(\cn{Mdm2})~\otimes~ 
      \cn{dont\_care}(\{\cn{p53},\cn{DNAdam}\}) \\
   \cn{not\_fireable}(3) \eqdef \cn{abs}(\cn{p53}) ~\otimes~ 
      \cn{dont\_care}(\{\cn{Mdm2},\cn{DNAdam}\}) \\
   \cn{not\_fireable}(4) \eqdef \cn{abs}(\cn{Mdm2}) ~\otimes~ 
      \cn{dont\_care}(\cn{\{p53,DNAdam\}}) \\
   \cn{not\_fireable}(5) \eqdef \cn{abs}(\cn{p53}) ~\otimes~ 
      \cn{dont\_care}(\{\cn{DNAdam},\cn{Mdm2}\}) \\
   \cn{not\_fireable}(6) \eqdef \cn{pres}(\cn{DNAdam}) ~\otimes~ 
      \cn{dont\_care}(\{\cn{Mdm2},\cn{p53}\}) 
$
\end{quote}
\end{itemize}

\newpage
\section{Example Proofs}
\label{appendix:example.proofs}
%
%

In most of the following proofs,
we omit the intuitionistic constant environment  $\mathop{\dag} system ~@~ 0$
(from which we can deduce $system ~@~ w$ at any state $w$).
We also omit trivial proofs (called ``hyp'') consisting in using hypothesis from the 
(intuitionistic or linear) context.

Note that several proofs implicitly use the $\cn{unchanged}$ predicate, 
which, for example, enable to pr`ove
$\cn{abs}(\cn{p53}) ~@~ w.1$ from $\cn{abs}(\cn{p53}) ~@~ w$ in the $P_1$ 
part of the proof of Property 1.  
%
We shall also sometimes need (and implicitly use) the $\cn{well\_defined}$ hypothesis 
$\forall a.~ \cn{pres}(a) \oplus \cn{abs}(a)$ 
(see the proof of Property 4).  


\subsection{Property 1}
\label{appendix:sec.property1}
As long as there is DNA damage, the above system can oscillate (with a
short period) from $state_0$ to $state_1$ and back again.  

From $state_0$ and $\cn{pres}(\cn{DNAdam})$ 
we get $\cn{abs}(\cn{p53})$, $\cn{abs}(\cn{Mdm2})$, and $\cn{pres}(\cn{DNAdam})$ by rule $1$.
Then $\cn{pres}(\cn{p53})$, $\cn{abs}(\cn{Mdm2})$, and $\cn{pres}(\cn{DNAdam})$ ($state_1$) by rule $2$.
Then $\cn{pres}(\cn{p53})$, $\cn{pres}(\cn{Mdm2})$, and $\cn{pres}(\cn{DNAdam})$ by rule $3$,
and finally $\cn{abs}(\cn{p53})$, $\cn{pres}(\cn{Mdm2})$, and $\cn{pres}(\cn{DNAdam})$ ($state_0$) by rule $4$.

\begin{hyllprop}[Property 1, Version 1]
\label{appendix:prop:property1.1}  
For any world $w$, 
there exists two worlds $u$ and $v$ such that both $u$ and $v$ are less than $3$ and 
the following holds: 
$$ \begin{array}{l}
   \mathop{\dag} system ~@~ 0 ~;~ state_0 ~\otimes~ \cn{pres}(\cn{DNAdam}) ~@~ w \\
   \quad \vdashseq  
      \delay{u}~ 
      [( state_1 ~\otimes~ \cn{dont\_care}(\cn{DNAdam}) ) ~\mathbin{\&}~ 
       ( \delay{v}~ (state_0 ~\otimes~ \cn{dont\_care}(\cn{DNAdam})) )] ~@~ w 
   \end{array}
$$
\end{hyllprop}
%
\begin{proof} 

Let $S_0$ denote $state_0 ~\otimes~ \cn{pres}(\cn{DNAdam})$ 
where $state_0$ is $\cn{abs}(\cn{p53}) \otimes \cn{pres}(\cn{Mdm2})$

in
$$\begin{array}{c}
\hspace{-0.5cm}
\dfrac
  {\dfrac
     {\dfrac
         {P_1}
         {S_0 ~@~ w \vdashseq state_1 ~\otimes~ \cn{dont\_care}(\cn{DNAdam})  ~@~ w.u}
     \qquad
     \dfrac
        {\dfrac
           {P_2}
           {S_0 ~@~ w \vdashseq state_0 ~\otimes~ \cn{dont\_care}(\cn{DNAdam})  ~@~ w.u.v}
        }
        {S_0 ~@~ w \vdashseq \delay{v}~ (state_0 ~\otimes~ \cn{dont\_care}(\cn{DNAdam}))  ~@~ w.u}
    }
     {S_0 ~@~ w \vdashseq 
        ( state_1 ~\otimes~ \cn{dont\_care}(\cn{DNAdam}) ) ~\mathbin{\&}~ 
        \delay{v}~ (state_0 ~\otimes~ \cn{dont\_care}(\cn{DNAdam})) ~@~ w.u } ~\mathbin{\&} R
  }
  {S_0 ~@~ w \vdashseq \delay{u}~ 
     [( state_1 ~\otimes~ \cn{dont\_care}(\cn{DNAdam}) ) ~\mathbin{\&}~ 
      ( \delay{v}~ (state_0 ~\otimes~ \cn{dont\_care}(\cn{DNAdam})) )] ~@~ w }
\end{array}$$
where $P_1$ is

let us remind that

 $\begin{array}{ll}
  \cn{rule}(1) & \eqdef \cn{inhib}(\cn{vars},\cn{DNAdam},\cn{Mdm2}) \\
         & \eqdef (\cn{pres}(\cn{DNAdam}) 
             \oplus (\cn{pres}(\cn{DNAdam})  \otimes \cn{pres}(\cn{Mdm2})) 
             \oplus (\cn{pres}(\cn{DNAdam})  \otimes \cn{abs}(\cn{Mdm2}))) \\
          & \quad ~ \limp \delay{1} (\cn{pres}(\cn{DNAdam}) \otimes \cn{abs}(\cn{Mdm2})) ~ \otimes \downarrow u.~ 
             \cn{unchanged}(\cn{p53},u),
  \end{array}$

let $H_1$ denote 
   $\cn{pres}(\cn{DNAdam}) 
      \oplus (\cn{pres}(\cn{DNAdam})  \otimes \cn{pres}(\cn{Mdm2})) 
      \oplus (\cn{pres}(\cn{DNAdam})  \otimes \cn{abs}(\cn{Mdm2}))$ (hyp. of $\cn{rule}(1)$),  

$C_1$ denote
   $\delay{1} (\cn{pres}(\cn{DNAdam}) \otimes \cn{abs}(\cn{Mdm2})) ~ 
      \otimes \downarrow u.~ \cn{unchanged}(\cn{p53},u)$ (conclusion of $\cn{rule}(1)$),

and 
$R_1$ denote $state_1 ~\otimes~ \cn{dont\_care}(\cn{DNAdam})$
where $state_1$ is $\cn{pres}(\cn{p53})$, $\cn{abs}(\cn{Mdm2})$

in
$$\begin{array}{c}
\dfrac
  {\dfrac
      { ... \vdashseq \cn{pres}(\cn{DNAdam}) \otimes \cn{pres}(\cn{Mdm2}) ~@~ w}
      {\cn{pres}(\cn{Mdm2}) \otimes \cn{pres}(\cn{DNAdam}) ~@~ w \vdashseq H_1 ~@~ w} ~ \oplus R_1
   \qquad
   \dfrac
      {P_{11}}
      {\cn{abs}(\cn{p53})  ~@~w,~ C_1 ~@~ w \vdashseq R_1 ~@~ w.u}
 }
  {\cn{abs}(\cn{p53}) \otimes \cn{pres}(\cn{Mdm2}) \otimes \cn{pres}(\cn{DNAdam}) ~@~ w, \cn{rule}(1) ~@~ w 
     \vdashseq state_1 ~\otimes~ \cn{dont\_care}(\cn{DNAdam}) ~@~ w.u} ~\limp L
\end{array}$$
with $P_{11}$ 
$$\begin{array}{c}
\hspace{-0.5cm}
\dfrac
  {\dfrac
      { P_{12} }
      {\cn{abs}(\cn{p53}) \otimes \cn{abs}(\cn{Mdm2}) \otimes \cn{pres}(\cn{DNAdam}) ~@~ w.1  \vdashseq R_1 ~@~ w.u}
  }
  {\cn{abs}(\cn{p53}) ~@~w,~ \cn{pres}(\cn{DNAdam}) \otimes \cn{abs}(\cn{Mdm2}) ~@~ w.1,~  \cn{unchanged}(\cn{p53},w) ~@~ w 
      \vdashseq R_1 ~@~ w.u} 
\end{array}$$

\newpage

$P_{12}$: 

let us remind that

 $\begin{array}{ll}
      \cn{rule}(2) & \eqdef \cn{inhib_s}(\cn{vars},\cn{Mdm2},\cn{p53}) \\
         & \eqdef (\cn{abs}(\cn{Mdm2}) 
             \oplus (\cn{abs}(\cn{Mdm2}) \otimes \cn{pres}(\cn{p53})) 
             \oplus (\cn{abs}(\cn{Mdm2}) \otimes \cn{abs}(\cn{p53}))) \\
         & \quad ~   \limp \delay{1} (\cn{abs}(\cn{Mdm2}) \otimes \cn{pres}(\cn{p53})) ~ \otimes \downarrow u.~ 
             \cn{unchanged}(\cn{DNAdam},u)),
  \end{array}$

let $H_2$ denote 
   $\cn{abs}(\cn{Mdm2}) 
             \oplus (\cn{abs}(\cn{Mdm2}) \otimes \cn{pres}(\cn{p53})) 
             \oplus (\cn{abs}(\cn{Mdm2}) \otimes \cn{abs}(\cn{p53}))$ (hypothesis of $\cn{rule}(2)$),

and $C_2$ denote
   $\delay{1} (\cn{abs}(\cn{Mdm2}) \otimes \cn{pres}(\cn{p53})) ~ \otimes \downarrow u.~ 
             \cn{unchanged}(\cn{DNAdam},u))$ (conclusion of $\cn{rule}(2)$),

in
$$\begin{array}{c}
\dfrac
  {\dfrac
      { ... \vdashseq \cn{abs}(\cn{p53}) \otimes \cn{abs}(\cn{Mdm2}) ~@~ w.1} 
      {\cn{abs}(\cn{p53}) \otimes \cn{abs}(\cn{Mdm2}) ~@~ w.1 \vdashseq H_2 ~@~ w.1}
   \quad
   P_{12c}
  }
  {\cn{abs}(\cn{p53}) \otimes \cn{abs}(\cn{Mdm2}) \otimes \cn{pres}(\cn{DNAdam}) ~@~ w.1,~ \cn{rule}(2) ~@~ w.1  
    \vdashseq R_1 ~@~ w.u} ~\limp L
\end{array}$$

with $P_{12c}$:
$$\begin{array}{c}
   \dfrac
      {\dfrac
         { \dfrac
           { state_1 ~@~ w.2 \vdashseq state_1 ~@~ w.2
             \quad
            \cn{pres}(\cn{DNAdam})  ~@~ w.2 \vdashseq \cn{dont\_care}(\cn{DNAdam}) ~@~ w.2 ~[u=2]
           }
           { \cn{abs}(\cn{Mdm2}) \otimes \cn{pres}(\cn{p53}) \otimes \cn{pres}(\cn{DNAdam})  ~@~ w.2 \vdashseq R_1 ~@~ w.u  } ~ \otimes R
         } 
        {\cn{abs}(\cn{Mdm2}) \otimes \cn{pres}(\cn{p53}) ~@~ w.2,~  
          \cn{pres}(\cn{DNAdam}) \otimes \cn{unchanged}(\cn{DNAdam},w.1) ~@~ w.1 
            \vdashseq R_1 ~@~ w.u }
      }
      {\cn{pres}(\cn{DNAdam}) \otimes C_2 ~@~ w.1 \vdashseq R_1 ~@~ w.u} 
\end{array}$$
%

and $P_2$ is

($P_2$ is $P_1$ where $R_1 ~@~ w.u$ is replaced by $R_2 ~@~ w.u.v$ and $P_{12c}$ is replaced by $P_{23}$.)

let $H_1$ denote 
   $\cn{pres}(\cn{DNAdam}) 
      \oplus (\cn{pres}(\cn{DNAdam})  \otimes \cn{pres}(\cn{Mdm2})) 
      \oplus (\cn{pres}(\cn{DNAdam})  \otimes \cn{abs}(\cn{Mdm2}))$ (hyp. of $\cn{rule}(1)$),  

$C_1$ denote
   $\delay{1} (\cn{pres}(\cn{DNAdam}) \otimes \cn{abs}(\cn{Mdm2})) ~ 
      \otimes \downarrow u.~ \cn{unchanged}(\cn{p53},u)$ (conclusion of $\cn{rule}(1)$),

and 
$R_2$ denote $state_0 ~\otimes~ \cn{dont\_care}(\cn{DNAdam})$
where $state_p$ is $\cn{abs}(\cn{p53}) \otimes \cn{pres}(\cn{Mdm2})$

in

$$\begin{array}{c}
\dfrac
  {\dfrac
      { ... \vdashseq \cn{pres}(\cn{DNAdam}) \otimes \cn{pres}(\cn{Mdm2}) ~@~ w}
      {\cn{pres}(\cn{Mdm2}) \otimes \cn{pres}(\cn{DNAdam}) ~@~ w \vdashseq H_1 ~@~ w} ~ \oplus R_1
   \qquad
   \dfrac
      {P_{21}}
      {\cn{abs}(\cn{p53})  ~@~w,~ C_1 ~@~ w \vdashseq R_2 ~@~ w.u.v}
 }
  {\cn{abs}(\cn{p53}) \otimes \cn{pres}(\cn{Mdm2}) \otimes \cn{pres}(\cn{DNAdam}) ~@~ w, \cn{rule}(1) ~@~ w 
     \vdashseq state_0 ~\otimes~ \cn{dont\_care}(\cn{DNAdam}) ~@~ w.u} ~\limp L
\end{array}$$
with $P_{21}$ 
$$\begin{array}{c}
\hspace{-0.5cm}
\dfrac
  {\dfrac
      { P_{22} }
      {\cn{abs}(\cn{p53}) \otimes \cn{abs}(\cn{Mdm2}) \otimes \cn{pres}(\cn{DNAdam}) ~@~ w.1  \vdashseq R_2 ~@~ w.u.v}
  }
  {\cn{abs}(\cn{p53}) ~@~w,~ \cn{pres}(\cn{DNAdam}) \otimes \cn{abs}(\cn{Mdm2}) ~@~ w.1,~  \cn{unchanged}(\cn{p53},w) ~@~ w 
      \vdashseq R_2 ~@~ w.u.v} 
\end{array}$$
$P_{22}$: 

let $H_2$ denote 
   $(\cn{abs}(\cn{Mdm2}) 
             \oplus (\cn{abs}(\cn{Mdm2}) \otimes \cn{pres}(\cn{p53})) 
             \oplus (\cn{abs}(\cn{Mdm2}) \otimes \cn{abs}(\cn{p53})))$ (hypothesis of $\cn{rule}(2)$),

and $C_2$ denote
   $\delay{1} (\cn{abs}(\cn{Mdm2}) \otimes \cn{pres}(\cn{p53})) ~ \otimes \downarrow u.~ 
             \cn{unchanged}(\cn{DNAdam},u))$ (conclusion of $\cn{rule}(2)$),

in
$$\begin{array}{c}
\dfrac
  {\dfrac
      { ... \vdashseq \cn{abs}(\cn{p53}) \otimes \cn{abs}(\cn{Mdm2}) ~@~ w.1} 
      {\cn{abs}(\cn{p53}) \otimes \cn{abs}(\cn{Mdm2}) ~@~ w.1 \vdashseq H_2 ~@~ w.1}
   \quad
   P_{22c}
  }
  {\cn{abs}(\cn{p53}) \otimes \cn{abs}(\cn{Mdm2}) \otimes \cn{pres}(\cn{DNAdam}) ~@~ w.1,~ \cn{rule}(2) ~@~ w.1  
    \vdashseq R_2 ~@~ w.u.v} ~\limp L
\end{array}$$

with $P_{22c}$:
$$\begin{array}{c}
   \dfrac
      {\dfrac
         { \dfrac
           { P_{23}
           }
           {\cn{abs}(\cn{Mdm2}) \otimes \cn{pres}(\cn{p53}) \otimes \cn{pres}(\cn{DNAdam})  ~@~ w.2 \vdashseq R_2 ~@~ w.u.v} 
         } 
        {\cn{abs}(\cn{Mdm2}) \otimes \cn{pres}(\cn{p53}) ~@~ w.2,~  
          \cn{pres}(\cn{DNAdam}) \otimes \cn{unchanged}(\cn{DNAdam},w.1) ~@~ w.1 
            \vdashseq R_2 ~@~ w.u.v}
      }
      {\cn{pres}(\cn{DNAdam}) \otimes C_2 ~@~ w.1 \vdashseq R_2 ~@~ w.u} 
\end{array}$$

\newpage

$P_{23}$:

let us remind that

  $\begin{array}{ll}
     \cn{rule}(3) & \eqdef \cn{active}(\cn{vars},\cn{p53},\cn{Mdm2}) \\
        & \eqdef (\cn{pres}(\cn{p53}) 
            \oplus (\cn{pres}(\cn{p53}) \otimes \cn{pres}(\cn{Mdm2})) 
            \oplus (\cn{pres}(\cn{p53}) \otimes \cn{abs}(\cn{Mdm2}))) \\
        & \quad ~   \limp \delay{1} (\cn{pres}(\cn{p53}) \otimes \cn{pres}(\cn{Mdm2})) ~ \otimes \downarrow u.~ 
           \cn{unchanged}(\cn{DNAdam},u)),
  \end{array}$

let $H_3$ denote $\cn{pres}(\cn{p53}) 
            \oplus (\cn{pres}(\cn{p53}) \otimes \cn{pres}(\cn{Mdm2})) 
            \oplus (\cn{pres}(\cn{p53}) \otimes \cn{abs}(\cn{Mdm2}))$ (hyp. of $\cn{rule}(3)$),

and $C_3$ denote 
   $\delay{1} (\cn{pres}(\cn{p53}) \otimes \cn{pres}(\cn{Mdm2})) ~ \otimes 
    \downarrow u.~ \cn{unchanged}(\cn{DNAdam},u)) $ (conclusion of $\cn{rule}(3)$) 

in
$$\begin{array}{c}
\hspace{-0.5cm}
\dfrac
  {\dfrac
      { ... \vdashseq \cn{pres}(\cn{p53}) \otimes \cn{abs}(\cn{Mdm2})  ~@~ w.2 } 
      { \cn{abs}(\cn{Mdm2}) \otimes \cn{pres}(\cn{p53}) ~@~ w.2 \vdashseq H_3 ~@~ w.2}
     \qquad
   \dfrac
      {\dfrac
        {P_{24}}
        { \cn{pres}(\cn{p53}) \otimes \cn{pres}(\cn{Mdm2}) 
         \otimes  \cn{pres}(\cn{DNAdam}) ~@~ w.3 \vdashseq R_2 ~@~ w.u.v}
      }
      { \cn{pres}(\cn{DNAdam})  ~@~ w.2,~ C_3 ~@~ w.2 \vdashseq R_2 ~@~ w.u.v}
   }
  {\cn{abs}(\cn{Mdm2}) \otimes \cn{pres}(\cn{p53}) \otimes \cn{pres}(\cn{DNAdam})  ~@~ w.2,~ \cn{rule}(3) ~@~ w.2
   \vdashseq R_2 ~@~ w.u.v} 
\end{array}$$

$P_{24}$:

let us remind that

$\begin{array}{ll}
      \cn{rule}(4) & \eqdef \cn{inhib}(\cn{vars},\cn{Mdm2},\cn{p53}) \\
        & \eqdef (\cn{pres}(\cn{Mdm2})  
            \oplus (\cn{pres}(\cn{Mdm2}) \otimes \cn{pres}(\cn{p53})) 
            \oplus (\cn{pres}(\cn{Mdm2}) \otimes \cn{abs}(\cn{p53}))) \\
         & \limp \delay{1} (\cn{pres}(\cn{Mdm2}) \otimes \cn{abs}(\cn{p53})) ~ \otimes \downarrow u.~ 
             \cn{unchanged}(\cn{DNAdam},u))
\end{array}$,

and $R_2$ is $state_0 ~\otimes~ \cn{dont\_care}(\cn{DNAdam})$
where $state_p$ is $\cn{abs}(\cn{p53}) \otimes \cn{pres}(\cn{Mdm2})$,

let $H_4$ denote $ \cn{pres}(\cn{Mdm2})  
            \oplus (\cn{pres}(\cn{Mdm2}) \otimes \cn{pres}(\cn{p53})) 
            \oplus (\cn{pres}(\cn{Mdm2}) \otimes \cn{abs}(\cn{p53})) $ (hyp. of $\cn{rule}(4)$),

and $C_4$ denote 
   $\delay{1} (\cn{pres}(\cn{Mdm2}) \otimes \cn{abs}(\cn{p53})) ~ \otimes 
    \downarrow u.~ \cn{unchanged}(\cn{DNAdam},u)) $ (conclusion of $\cn{rule}(4)$) 

in
$$\begin{array}{c}
\hspace{-1.0cm}
\dfrac
  {\dfrac
      { ... \vdashseq \cn{pres}(\cn{Mdm2}) \otimes \cn{pres}(\cn{p53})  ~@~ w.3 } 
      { \cn{pres}(\cn{p53}) \otimes \cn{pres}(\cn{Mdm2}) ~@~ w.3 \vdashseq H_4 ~@~ w.3}
     \quad
   \dfrac
      {\cn{pres}(\cn{Mdm2}) \otimes \cn{abs}(\cn{p53}) 
         \otimes  \cn{pres}(\cn{DNAdam}) ~@~ w.4 \vdashseq R_2 ~@~ w.u.v ~ [v=2]}
      { \cn{pres}(\cn{DNAdam}) ~@~ w.3,~ C_4 ~@~ w.3 \vdashseq R_2 ~@~ w.u.v}
   }
  {\cn{pres}(\cn{p53}) \otimes \cn{pres}(\cn{Mdm2}) \otimes \cn{pres}(\cn{DNAdam}) ~@~ w.3,~ \cn{rule}(4) ~@~ w.3
   \vdashseq R_2 ~@~ w.u.v ~~ [u=2; v=2]} 
\end{array}$$
\end{proof}

\begin{hyllprop}[Property 1, Version 2]
\label{appendix:prop:property1.2}  
For any world $w$, 
there exists two worlds $u$ and $v$ such that both $u$ and $v$ are less than $3$ and the following holds: 
$$
\begin{array}{l}
   \mathop{\dag} system ~@~ 0 ~;~ state_0 ~\otimes~ \cn{pres}(\cn{DNAdam})~@~ w 
      \vdashseq  state_1 ~\otimes~ \cn{dont\_care}(\cn{DNAdam}) ~@~ w.u ~~(1)
~~\textrm{and}~ \\
   \mathop{\dag} system ~@~ 0 ~;~ state_1 ~@~ w.u
      \vdashseq state_0 ~@~ w.u.v ~~(2)
\end{array}
$$
\end{hyllprop}
%

Note.
There are no $\cn{dont\_care}$'s needed in the conclusion of the second sequent 
because only rules 3 and 4 are used, which don't involve DNAdam.

\paragraph{}
Let us first prove the first statement (1):

\begin{proof} 

%
%
$$\begin{array}{c}
\dfrac
   {P_1}
   {state_0 ~\otimes~ \cn{pres}(\cn{DNAdam}) ~@~ w \vdashseq state_1 ~\otimes~ \cn{dont\_care}(\cn{DNAdam})  ~@~ w.u ~~[u=2]}
\end{array}$$
where $P_1$ is the same proof $P_1$ used in the proof of Property 1, Version 1, above.
\end{proof} 

Let us prove the second statement (2):

\begin{proof} 

let us remind that

  $\begin{array}{ll}
     \cn{rule}(3) & \eqdef \cn{active}(\cn{vars},\cn{p53},\cn{Mdm2}) \\
        & \eqdef (\cn{pres}(\cn{p53}) 
            \oplus (\cn{pres}(\cn{p53}) \otimes \cn{pres}(\cn{Mdm2})) 
            \oplus (\cn{pres}(\cn{p53}) \otimes \cn{abs}(\cn{Mdm2}))) \\
        & \quad ~   \limp \delay{1} (\cn{pres}(\cn{p53}) \otimes \cn{pres}(\cn{Mdm2})) ~ \otimes \downarrow u.~ 
           \cn{unchanged}(\cn{DNAdam},u)),
  \end{array}$

let $H_3$ denote $\cn{pres}(\cn{p53}) 
            \oplus (\cn{pres}(\cn{p53}) \otimes \cn{pres}(\cn{Mdm2})) 
            \oplus (\cn{pres}(\cn{p53}) \otimes \cn{abs}(\cn{Mdm2}))$ (hyp. of $\cn{rule}(3)$),

and $C_3$ denote 
   $\delay{1} (\cn{pres}(\cn{p53}) \otimes \cn{pres}(\cn{Mdm2})) ~ \otimes 
    \downarrow u.~ \cn{unchanged}(\cn{DNAdam},u)) $ (conclusion of $\cn{rule}(3)$) 

in
$$\begin{array}{c}
\dfrac
   {\dfrac
  {\dfrac
      { ... \vdashseq \cn{pres}(\cn{p53}) \otimes \cn{abs}(\cn{Mdm2})  ~@~ w.2 } 
      { \cn{abs}(\cn{Mdm2}) \otimes \cn{pres}(\cn{p53}) ~@~ w.2 \vdashseq H_3 ~@~ w.2}
     \qquad
   \dfrac
      {\dfrac
        {P_{24}}
        { \cn{pres}(\cn{p53}) \otimes \cn{pres}(\cn{Mdm2}) ~@~ w.3 
          \vdashseq state_0 ~@~ w.2.v}
      }
      {C_3 ~@~ w.2 \vdashseq state_0 ~@~ w.2.v}
   }
  {\cn{abs}(\cn{Mdm2}) \otimes \cn{pres}(\cn{p53}) ~@~ w.2,~ \cn{rule}(3) ~@~ w.2
   \vdashseq state_0 ~@~ w.2.v} 
   }
   {state_1 ~@~ w.2 \vdashseq state_0 ~@~ w.2.v ~~[v=2]}
\end{array}$$

$P_{24}$:

let us remind that

$\begin{array}{ll}
      \cn{rule}(4) & \eqdef \cn{inhib}(\cn{vars},\cn{Mdm2},\cn{p53}) \\
        & \eqdef (\cn{pres}(\cn{Mdm2})  
            \oplus (\cn{pres}(\cn{Mdm2}) \otimes \cn{pres}(\cn{p53})) 
            \oplus (\cn{pres}(\cn{Mdm2}) \otimes \cn{abs}(\cn{p53}))) \\
         & \limp \delay{1} (\cn{pres}(\cn{Mdm2}) \otimes \cn{abs}(\cn{p53})) ~ \otimes \downarrow u.~ 
             \cn{unchanged}(\cn{DNAdam},u)),
\end{array}$

let $H_4$ denote 
   $\cn{pres}(\cn{Mdm2})  
    \oplus (\cn{pres}(\cn{Mdm2}) \otimes \cn{pres}(\cn{p53})) 
    \oplus (\cn{pres}(\cn{Mdm2}) \otimes \cn{abs}(\cn{p53})) $ (hyp. of $\cn{rule}(4)$),

and $C_4$ denote 
   $\delay{1} (\cn{pres}(\cn{Mdm2}) \otimes \cn{abs}(\cn{p53})) ~ \otimes 
    \downarrow u.~ \cn{unchanged}(\cn{DNAdam},u)) $ (conclusion of $\cn{rule}(4)$) 

in
$$\begin{array}{c}
\dfrac
  {\dfrac
      { ... \vdashseq \cn{pres}(\cn{Mdm2}) \otimes \cn{pres}(\cn{p53})  ~@~ w.3 } 
      { \cn{pres}(\cn{p53}) \otimes \cn{pres}(\cn{Mdm2}) ~@~ w.3 \vdashseq H_4 ~@~ w.3}
     \quad
   \dfrac
      {\cn{pres}(\cn{Mdm2}) \otimes \cn{abs}(\cn{p53}) ~@~ w.4 
       \vdashseq state_0 ~@~ w.2.v ~ [v=2]}
      {C_4 ~@~ w.3 \vdashseq state_0 ~@~ w.2.v}
   }
  {\cn{pres}(\cn{p53}) \otimes \cn{pres}(\cn{Mdm2}) ~@~ w.3,~ \cn{rule}(4) ~@~ w.3
   \vdashseq state_0 ~@~ w.2.v ~~ [v=2]} 
\end{array}$$
\end{proof}

\subsection{Property 2}
\label{appendix:sec.property2}

DNA damage can be quickly recovered.

From $state_0$ and $\cn{pres}(\cn{DNAdam})$ 
we get $\cn{abs}(\cn{p53})$, $\cn{abs}(\cn{Mdm2})$, and $\cn{pres}(\cn{DNAdam})$ by rule $1$.
Then $\cn{pres}(\cn{p53})$ and $\cn{abs}(\cn{Mdm2})$ ($state_1$) and $\cn{pres}(\cn{DNAdam})$ by rule $2$.
Then $\cn{abs}(\cn{p53})$, $\cn{abs}(\cn{Mdm2})$, and $\cn{abs}(\cn{DNAdam})$ by rule $5$,
and finally $\cn{abs}(\cn{p53})$ and $\cn{pres}(\cn{Mdm2})$ ($state_0$) and $\cn{abs}(\cn{DNAdam})$ 
by rule $6$.

\begin{proposition}[Property 2]
\label{appendix:prop.property2}  
For any world $w$, 
there exists a world $u$ such that $u$ is less than $5$ and the following holds:
$$
\mathop{\dag} system ~@~ 0;~ 
    state_0 ~\otimes~ \cn{pres}(\cn{DNAdam}) ~@~ w
    \vdashseq state_0 \otimes \cn{abs}(\cn{DNAdam}) ~@~ w.u
$$
\end{proposition}

\begin{proof}

Let $R$ denote $state_0 \otimes \cn{abs}(\cn{DNAdam}) ~@~ w.4$
where $state_0$ is $\cn{abs}(\cn{p53}) \otimes \cn{pres}(\cn{Mdm2})$
in
$$\begin{array}{c}
\hspace{-0.3cm}
\dfrac
   {\dfrac
         {P_1 ~~
          \dfrac
            {\dfrac
                {P_2 ~~
                 \dfrac
                    {\dfrac
                       {P_5 ~~ 
                        \dfrac
                          {\dfrac
                              {P_6 ~~
                               \dfrac
                                 {state_0 \otimes \cn{abs}(\cn{DNAdam}) ~@~ w.4
                                  \vdashseq  state_0 \otimes \cn{abs}(\cn{DNAdam}) ~@~ w.4 ~~ [hyp]}
                                 {\cn{abs}(\cn{p53}) ~@~ w.3, ~
                                   \cn{pres}(\cn{Mdm2}) \otimes  \cn{abs}(\cn{DNAdam}) ~@~ w.4,~
                                   \cn{unchanged}(\cn{p53},w.3) ~@~ w.3 \vdashseq R }
                               }
                               {\cn{abs}(\cn{p53}) \otimes \cn{abs}(\cn{Mdm2}) \otimes \cn{abs}(\cn{DNAdam}) ~@~ w.3, 
                                \cn{rule}(6) ~@~ w.3 
                                \vdashseq R } 
                          }
                          {\cn{abs}(\cn{Mdm2}) ~@~ w.2,~
                           \cn{abs}(\cn{p53}) \otimes \cn{abs}(\cn{DNAdam}) ~@~ w.3,~
                           \cn{unchanged}(\cn{Mdm2},w.2) ~@~ w.2 \vdashseq R}
                       }
                       {\cn{pres}(\cn{p53}) \otimes \cn{abs}(\cn{Mdm2}) \otimes \cn{pres}(\cn{DNAdam})  ~@~ w.2,~ 
                       \cn{rule}(5) ~@~ w.2 \vdashseq R} 
                   }
                   {\cn{pres}(\cn{DNAdam})  ~@~ w.1, ~
                     \cn{pres}(\cn{p53}) \otimes \cn{abs}(\cn{Mdm2}) ~@~ w.2,~
                     \cn{unchanged}(\cn{DNAdam},w.1)  ~@~ w.1 \vdashseq R}
                }
                {\cn{abs}(\cn{p53}) \otimes \cn{abs}(\cn{Mdm2}) \otimes \cn{pres}(\cn{DNAdam}) ~@~ w.1,~ 
                 \cn{rule}(2) ~@~ w.1 \vdashseq R} 
            }
            {\cn{abs}(\cn{p53})  ~@~ w,~
             \cn{abs}(\cn{Mdm2}) \otimes \cn{pres}(\cn{DNAdam}) ~@~ w.1,~
              \cn{unchanged}(\cn{p53},w) ~@~ w \vdashseq R}
        }
         {state_0 ~\otimes~ \cn{pres}(\cn{DNAdam}) ~@~ w,~ \cn{rule}(1) ~@~ w
             \vdashseq R  
         }  \limp L 
   }
   {
    state_0 ~\otimes~ \cn{pres}(\cn{DNAdam}) ~@~ w
    \vdashseq state_0 \otimes \cn{abs}(\cn{DNAdam}) ~@~ w.4} 
\end{array}$$
where $P_1$ is
$$ \cn{pres}(\cn{Mdm2}) ~@~ w,~ \cn{pres}(\cn{DNAdam}) ~@~ w 
     \vdashseq  \cdots \oplus (\cn{pres}(\cn{Mdm2}) \otimes \cn{pres}(\cn{DNAdam})) \oplus \cdots  ~@~ w,
$$
$P_2$ is
$$ \cn{abs}(\cn{p53}) ~@~ w.1,~ \cn{abs}(\cn{Mdm2}) ~@~ w.1 
     \vdashseq  (\cn{abs}(\cn{p53}) \otimes \cn{abs}(\cn{Mdm2})) \oplus \cdots \oplus \cdots ~@~ w.1,
$$
$P_5$ is
$$ \cn{pres}(\cn{p53}) ~@~ w.2,~ \cn{pres}(\cn{DNAdam}) ~@~ w.2 
     \vdashseq  \cdots \oplus (\cn{pres}(\cn{p53}) \otimes \cn{pres}(\cn{DNAdam})) \oplus \cdots  ~@~ w.2,
$$
and $P_6$ is
$$ \cn{abs}(\cn{Mdm2}) ~@~ w.3,~ \cn{abs}(\cn{DNAdam}) ~@~ w.3 
     \vdashseq  (\cn{abs}(\cn{Mdm2}) \otimes \cn{abs}(\cn{DNAdam})) \oplus \cdots \oplus \cdots ~@~ w.3
$$
\end{proof}



\subsection{Property 3}
\label{appendix:sec.property3}
If there is no DNA damage, the system remains in the initial state.

\begin{hyllprop}[Property 3]
\label{appendix:prop:property3}
Let $\mathcal{P}$ denote the formula $state_0  ~\otimes~ \cn{abs}(\cn{DNAdam})$.
For any world $w$, the following holds:
$$\mathop{\dag} system ~@~ 0;~ \mathcal{P} ~@~ 0~  
  \vdashseq \mathcal{P} ~\at~ 0  ~@~ w;$$
and for any world $w$, for any rule $r$ in the interval $[1..6]$, 
the following holds: \\
$$\begin{array}{l}
\mathop{\dag} system ~@~ 0;~ . \vdashseq 
   \mathcal{P} \limp 
   (\cn{fireable}(r) \mathbin{\&} \delay{1} \mathcal{P})  ~\oplus~
     \cn{not\_fireable}(r)  ~@~ w
\end{array}$$
%
\end{hyllprop}

%
The proof of the second statement proceeds by case analysis on the rules ($r$) of the biological system.
There are only two fireable rules: $\cn{rule}(4)$ and $\cn{rule}(6)$.
 
%
\begin{proof}

Let 
$\mathcal{P}$ denote 
$state_0  ~\otimes~ \cn{abs}(\cn{DNAdam})$
where $state_0$ is $\cn{abs}(\cn{p53}) \otimes \cn{pres}(\cn{Mdm2})$

in

$$\begin{array}{c}
  \dfrac  
      {\dfrac
         {P_1 \quad P_2 \quad P_3 \quad P_4 \quad P_5 \quad P_6}  
         {\mathcal{P}  ~@~ w \vdashseq  \forall r: [1..6].~  
          (\cn{fireable}(r) \mathbin{\&} \delay{1} \mathcal{P})  ~\oplus~ \cn{not\_fireable}(r)  ~@~ w}~ case_r
      }
      { \vdashseq \mathcal{P} \limp \forall r: [1..6].~  
        (\cn{fireable}(r) \mathbin{\&} \delay{1} \mathcal{P})  ~\oplus~ \cn{not\_fireable}(r)  ~@~ w} ~\limp R
\end{array}$$

where $P_1$ is
\ednote[]{$
\cn{fireable}(1) \eqdef ( \cn{pres}(\cn{DNAdam}) 
             \oplus (\cn{pres}(\cn{DNAdam})  \otimes \cn{pres}(\cn{Mdm2}))
             \oplus (\cn{pres}(\cn{DNAdam})  \otimes \cn{abs}(\cn{Mdm2}))) 
             ~\otimes~ \cn{dont\_care}(\cn{p53}) \\
\cn{not\_fireable}(1) \eqdef \cn{abs}(\cn{DNAdam}) ~\otimes~ 
      \cn{dont\_care}(\{\cn{Mdm2},\cn{p53}\}) \\
$}
$$\begin{array}{c}
\hspace{-0.5cm}
\dfrac
   {\dfrac
      {\cn{abs}(\cn{DNAdam}) ~@~ w \vdashseq \cn{abs}(\cn{DNAdam}) ~@~ w 
       \qquad
       \cn{abs}(\cn{p53}) \otimes \cn{pres}(\cn{Mdm2}) ~@~ w \vdashseq \cn{dont\_care}(\{\cn{Mdm2},\cn{p53}\}) ~@~ w 
     }
      {\mathcal{P}  ~@~ w \vdashseq \cn{not\_fireable}(1)  ~@~ w}~ \oplus R
   }
   {\mathcal{P}  ~@~ w \vdashseq (\cn{fireable}(1) \mathbin{\&} \delay{1} \mathcal{P})  ~\oplus~ \cn{not\_fireable}(1)  ~@~ w}
\end{array}$$

$P_2$ is
\ednote[]{$
     \cn{fireable}(2) \eqdef  ( \cn{abs}(\cn{Mdm2}) 
             \oplus (\cn{abs}(\cn{Mdm2}) \otimes \cn{pres}(\cn{p53})) 
             \oplus (\cn{abs}(\cn{Mdm2}) \otimes \cn{abs}(\cn{p53})) )
             ~\otimes~ \cn{dont\_care}(\cn{DNAdam}) \\
   \cn{not\_fireable}(2) \eqdef \cn{pres}(\cn{Mdm2}) ~\otimes~ 
      \cn{dont\_care}(\{\cn{p53},\cn{DNAdam}\}) \\
$}
%
{\small
$$\begin{array}{c}
\hspace{-0.7cm}
\dfrac
   {\dfrac
      {\cn{abs}(\cn{DNAdam}) ~@~ w \vdashseq \cn{dont\_care}(\cn{DNAdam}) ~@~ w 
       \qquad
       \cn{abs}(\cn{p53}) \otimes \cn{pres}(\cn{Mdm2}) ~@~ w \vdashseq \cn{pres}(\cn{Mdm2}) ~\otimes~ \cn{dont\_care}(\cn{p53}) ~@~ w 
     }
     {\mathcal{P}  ~@~ w \vdashseq \cn{not\_fireable}(2)  ~@~ w}~ \oplus R
   }
   {\mathcal{P}  ~@~ w \vdashseq  (\cn{fireable}(2) \mathbin{\&} \delay{1} \mathcal{P})  ~\oplus~ \cn{not\_fireable}(2)  ~@~ w}
\end{array}$$
}

$P_3$ is
\ednote[]{$
     \cn{fireable}(3) \eqdef  ( \cn{pres}(\cn{p53}) 
            \oplus (\cn{pres}(\cn{p53}) \otimes \cn{pres}(\cn{Mdm2})) 
            \oplus (\cn{pres}(\cn{p53}) \otimes \cn{abs}(\cn{Mdm2})) ) 
             ~\otimes~ \cn{dont\_care}(\cn{DNAdam}) \\
   \cn{not\_fireable}(3) \eqdef \cn{abs}(\cn{p53}) ~\otimes~ 
      \cn{dont\_care}(\{\cn{Mdm2},\cn{DNAdam}\}) \\
$}
%
{\small
$$\begin{array}{c}
\hspace{-0.5cm}
\dfrac
   {\dfrac
      {\cn{abs}(\cn{DNAdam}) ~@~ w \vdashseq \cn{dont\_care}(\cn{DNAdam}) ~@~ w 
       \qquad
       \cn{abs}(\cn{p53}) \otimes \cn{pres}(\cn{Mdm2}) ~@~ w \vdashseq \cn{abs}(\cn{p53}) ~\otimes~ \cn{dont\_care}(\cn{Mdm2}) ~@~ w 
     }
     {\mathcal{P}  ~@~ w \vdashseq \cn{not\_fireable}(3)  ~@~ w}~ \oplus R
   }
   {\mathcal{P}  ~@~ w \vdashseq  (\cn{fireable}(3) \mathbin{\&} \delay{1} \mathcal{P})  ~\oplus~ \cn{not\_fireable}(3)  ~@~ w}
\end{array}$$
}

$P_4$ is
\ednote[]{
$\begin{array}{ll}
     \cn{rule}(4) 
         \eqdef & (\cn{pres}(\cn{Mdm2})  
            \oplus (\cn{pres}(\cn{Mdm2}) \otimes \cn{pres}(\cn{p53})) 
            \oplus (\cn{pres}(\cn{Mdm2}) \otimes \cn{abs}(\cn{p53}))) \\
         & \limp \delay{1} (\cn{pres}(\cn{Mdm2}) \otimes \cn{abs}(\cn{p53})) ~ \otimes \downarrow u.~ 
             \cn{unchanged}(\cn{DNAdam},u)) \\
\end{array}$
}
\ednote[]{$\cn{fireable}(4) \eqdef  ( \cn{pres}(\cn{Mdm2})  
            \oplus (\cn{pres}(\cn{Mdm2}) \otimes \cn{pres}(\cn{p53})) 
            \oplus (\cn{pres}(\cn{Mdm2}) \otimes \cn{abs}(\cn{p53}))) 
            \otimes \cn{dont\_care}(\cn{DNAdam}) \\
     \cn{not\_fireable}(4) \eqdef \cn{abs}(\cn{Mdm2}) \otimes \cn{dont\_care}(\cn{\{p53,DNAdam\}})$
}
$$\begin{array}{c}
\hspace{-0.5cm}
\dfrac
   {\dfrac
      {P_{4f} \qquad P_{4r}}
      {\cn{abs}(\cn{p53}) \otimes \cn{pres}(\cn{Mdm2}) \otimes \cn{abs}(\cn{DNAdam})  ~@~ w  
       \vdashseq \cn{fireable}(4) \mathbin{\&} \delay{1} \mathcal{P} ~@~ w
      } \mathbin{\&} R
   }
   {\mathcal{P}  ~@~ w 
    \vdashseq ~(\cn{fireable}(4) \mathbin{\&} \delay{1} \mathcal{P}) \oplus \cn{not\_fireable}(4) ~@~ w} \oplus R_1
\end{array}$$
$P_{4f}$ :
{\small
$$\begin{array}{c}
\hspace{-0.4cm}
\dfrac
   {\cn{abs}(\cn{p53}) ~\otimes~ \cn{pres}(\cn{Mdm2}) ~@~ w   
    \vdashseq \cn{pres}(\cn{Mdm2}) ~\otimes~ \cn{abs}(\cn{p53}) ~@~ w  
    \qquad 
    \cn{abs}(\cn{DNAdam}) ~@~ n \vdashseq \cn{dont\_care}(\cn{DNAdam}) ~@~ w}
   {\cn{abs}(\cn{p53}) \otimes \cn{pres}(\cn{Mdm2}) \otimes \cn{abs}(\cn{DNAdam})  ~@~ n \vdashseq \cn{fireable}(4) ~@~ w}  ~\otimes R      
\end{array}$$
}
$P_{4r}$ :
{\footnotesize
$$\begin{array}{c}
\hspace{-0.5cm}
\dfrac
  {\begin{array}{l} \cn{abs}(\cn{p53}) \otimes \cn{pres}(\cn{Mdm2}) ~@~ w \\
       \vdashseq \cn{abs}(\cn{p53}) \otimes \cn{pres}(\cn{Mdm2}) ~@~ w 
   \end{array}
  ~~ 
   \dfrac
     {\cdots \vdashseq  \cn{abs}(\cn{DNAdam})) ~@~ w+1
     \qquad 
     \cdots \vdashseq \cn{abs}(\cn{p53}) ~\otimes~ \cn{pres}(\cn{Mdm2})  ~@~ w+1
     }
    {\cn{abs}(\cn{DNAdam}) ~@~ w,~
       \cn{pres}(\cn{Mdm2}) ~\otimes~ \cn{abs}(\cn{p53}) ~@~ w+1,~ 
       \cn{unchanged}(\cn{DNAdam} ) ~@~ w 
       \vdashseq \mathcal{P} ~@~ w+1 
    } 
  } 
  {s\_\cn{rule}(4) ~@~ w, 
   ~\cn{abs}(\cn{p53}) \otimes \cn{pres}(\cn{Mdm2}) \otimes \cn{abs}(\cn{DNAdam}) ~@~ w 
  \vdashseq \delay{1} \mathcal{P} ~@~ w
  } \limp L
\end{array}$$
}

$P_5$ is
\ednote[]{$
     \cn{fireable}(5) \eqdef ( \cn{pres}(\cn{p53})  
            \oplus (\cn{pres}(\cn{p53}) \otimes \cn{pres}(\cn{DNAdam})) 
            \oplus (\cn{pres}(\cn{p53}) \otimes \cn{abs}(\cn{DNAdam})) )
             ~\otimes~ \cn{dont\_care}(\cn{Mdm2}) \\
   \cn{not\_fireable}(5) \eqdef \cn{abs}(\cn{p53}) ~\otimes~ 
      \cn{dont\_care}(\{\cn{DNAdam},\cn{Mdm2}\}) \\
$}
%
{\small
$$\begin{array}{c}
\hspace{-0.5cm}
\dfrac
   {\dfrac
      {\cn{abs}(\cn{DNAdam}) ~@~ w \vdashseq \cn{dont\_care}(\cn{DNAdam}) ~@~ w 
       \qquad
       \cn{abs}(\cn{p53}) \otimes \cn{pres}(\cn{Mdm2}) ~@~ w \vdashseq \cn{abs}(\cn{p53}) ~\otimes~ \cn{dont\_care}(\cn{Mdm2}) ~@~ w 
     }
     {\mathcal{P}  ~@~ w \vdashseq \cn{not\_fireable}(5)  ~@~ w}~ \oplus R
   }
   {\mathcal{P}  ~@~ w \vdashseq  (\cn{fireable}(5) \mathbin{\&} \delay{1} \mathcal{P})  ~\oplus~ \cn{not\_fireable}(5)  ~@~ w}
\end{array}$$
}

$P_6$ is
\ednote[]{
$\begin{array}{ll} 
   \cn{rule}(6) \eqdef & 
      (\cn{abs}(\cn{DNAdam})  
      \oplus (\cn{abs}(\cn{DNAdam}) \otimes \cn{pres}(\cn{Mdm2})) 
      \oplus (\cn{abs}(\cn{DNAdam}) \otimes \cn{abs}(\cn{Mdm2}))) \\
      & \limp \delay{1} (\cn{abs}(\cn{DNAdam}) \otimes \cn{pres}(\cn{Mdm2})) ~ \otimes \downarrow u.~ 
             \cn{unchanged}(\cn{p53},u))
   \end{array}$}
\ednote[]{$
     \cn{fireable}(6) \eqdef ( \cn{abs}(\cn{DNAdam})  
             \oplus (\cn{abs}(\cn{DNAdam}) \otimes \cn{pres}(\cn{Mdm2})) 
             \oplus (\cn{abs}(\cn{DNAdam}) \otimes \cn{abs}(\cn{Mdm2})) )
             ~\otimes~ \cn{dont\_care}(\cn{p53})  \\
   \cn{not\_fireable}(6) \eqdef \cn{pres}(\cn{DNAdam}) ~\otimes~ 
      \cn{dont\_care}(\{\cn{Mdm2},\cn{p53}\}) 
$}
$$\begin{array}{c}
\dfrac
   {\dfrac
      {P_{6f} \qquad P_{6r}}
      {\cn{abs}(\cn{p53}) \otimes \cn{pres}(\cn{Mdm2}) \otimes \cn{abs}(\cn{DNAdam})  ~@~ w  
       \vdashseq \cn{fireable}(6) \mathbin{\&} \delay{1} \mathcal{P} ~@~ w
      } \mathbin{\&} R
   }
   {\mathcal{P}  ~@~ w 
    \vdashseq ~(\cn{fireable}(6) \mathbin{\&} \delay{1} \mathcal{P}) \oplus \cn{not\_fireable}(6) ~@~ w} \oplus R_1
\end{array}$$
$P_{6f}$ :
{\small
$$\begin{array}{c}
\dfrac
   {\cn{pres}(\cn{Mdm2}) ~\otimes~ \cn{abs}(\cn{DNAdam}) ~@~ w  
       \vdashseq \cn{abs}(\cn{DNAdam}) ~\otimes~ \cn{pres}(\cn{Mdm2}) ~@~ w   
    \qquad 
    \cn{abs}(\cn{p53}) ~@~ n \vdashseq \cn{dont\_care}(\cn{p53}) ~@~ w}
   {\cn{abs}(\cn{p53}) \otimes \cn{pres}(\cn{Mdm2}) \otimes \cn{abs}(\cn{DNAdam})  ~@~ n \vdashseq \cn{fireable}(6) ~@~ w}  ~\otimes R      
\end{array}$$
}
$P_{6r}$ :
{\footnotesize
$$\begin{array}{c}
\hspace{-0.5cm}
\dfrac
  {\begin{array}{l} \cn{pres}(\cn{Mdm2}) \otimes \cn{abs}(\cn{DNAdam}) ~@~ w \\
       \vdashseq \cn{abs}(\cn{DNAdam}) \otimes \cn{pres}(\cn{Mdm2}) ~@~ w 
   \end{array}
  ~~ 
   \dfrac
     {\cdots \vdashseq  \cn{abs}(\cn{DNAdam})) ~@~ w+1 
     \qquad 
     \cdots \vdashseq \cn{abs}(\cn{p53}) ~\otimes~ \cn{pres}(\cn{Mdm2})  ~@~ w+1
     }
     {\cn{abs}(\cn{p53}) ~@~ n,~
       \cn{abs}(\cn{DNAdam}) \otimes \cn{pres}(\cn{Mdm2}) ~@~ w+1,~ 
      \cn{unchanged}(\cn{p53} ) ~@~ w 
       \vdashseq \mathcal{P} ~@~ w+1 
    } 
  } 
  {s\_\cn{rule}(6) ~@~ w, 
   ~\cn{abs}(\cn{p53}) \otimes \cn{pres}(\cn{Mdm2}) \otimes \cn{abs}(\cn{DNAdam}) ~@~ w 
  \vdashseq \delay{1} \mathcal{P} ~@~ w
  } \limp L
\end{array}$$
}
\end{proof}


\subsection{Property 4}
\label{appendix:sec.property4}
There is no path with two consecutive states where $\cn{p53}$ and $\cn{Mdm2}$ are both present or both absent.
In other words: 
from any state where $\cn{p53}$ and $\cn{Mdm2}$ are both present or both absent, 
we can only go to a state where either $\cn{p53}$ is present and $\cn{Mdm2}$ is absent or 
$\cn{p53}$ is absent and $\cn{Mdm2}$ is present.
This property is only valid for strong rules.  

\begin{proposition}[Property 4]
\label{appendix:prop.property4} 
For any world $n$, the following holds:
$$\mathop{\dag} system ~@~ 0 ;~ . \vdashseq \mathcal{L} \limp 
  \forall r: [1..6].~ (s\_\cn{fireable}(r) \mathbin{\&} \delay{1} \mathcal{R}) \oplus s\_\cn{not\_fireable}(r) ~@~ n $$
\end{proposition}

The proof proceeds by case analysis on the rules ($r$) of the biological system.

Note that we need (and sometimes implicitly use) here the $\cn{well\_defined}$ hypothesis 
$ \forall a \in \cn{vars}.~(\cn{pres}(a) \oplus \cn{abs}(a))$. 
For example, the proof $P_{1}$ below uses the hypothesis
$\cn{pres}(\cn{DNAdam}) \oplus \cn{abs}(\cn{DNAdam})$.

\begin{proof}

Let 
$\mathcal{L}$ denote
$(\cn{pres}(\cn{p53}) ~\otimes~ \cn{pres}(\cn{Mdm2})) ~\oplus~ 
 (\cn{abs}(\cn{p53}) ~\otimes~ \cn{abs}(\cn{Mdm2}))$

and $\mathcal{R}$ be
$( (\cn{pres}(\cn{p53}) ~\otimes~ \cn{abs}(\cn{Mdm2})) ~\oplus~ 
    (\cn{abs}(\cn{p53}) ~\otimes~ \cn{pres}(\cn{Mdm2})) )
 ~\otimes \cn{dont\_care}(\cn{DNAdam})$

in
$$\begin{array}{c}
   \dfrac
      {\dfrac
          { P_1 \quad P_2 }
          {\mathcal{L} ~@~ n \vdashseq \forall r: [1..6]. ~ (s\_\cn{fireable}(r) \mathbin{\&} 
              \delay{1} \mathcal{R}) \oplus s\_\cn{not\_fireable}(r) ~@~ n
          } \oplus L
       }
       {\vdashseq \mathcal{L} \limp 
        \forall r: [1..6]. ~ (s\_\cn{fireable}(r) \mathbin{\&} \delay{1} \mathcal{R}) \oplus s\_\cn{not\_fireable}(r) ~@~ n} \limp R
\end{array}$$
where $P_1$ is
$$\begin{array}{c}
\dfrac
   {P_{11} \quad P_{12} \quad P_{13} \quad P_{14} \quad P_{15} \quad P_{16}}
   {(\cn{pres}(\cn{p53}) ~\otimes~ \cn{pres}(\cn{Mdm2}))  ~@~ n 
    \vdashseq \forall r: [1..6]. ~ (s\_\cn{fireable}(r) \mathbin{\&} \delay{1} \mathcal{R}) \oplus s\_\cn{not\_fireable}(r) ~@~ n} ~ case_r
\end{array}$$
with $P_{11}$ :
\ednote[]{$s\_\cn{rule}(1) \eqdef \cn{pres}(\cn{DNAdam}) \otimes \cn{pres}(\cn{Mdm2}) 
 \limp \delay{1} (\cn{pres}(\cn{DNAdam}) \otimes \cn{abs}(\cn{Mdm2}))
 \otimes \downarrow u. \cn{unchanged}(\cn{p53},u)$}
 \ednote[]{
     $s\_\cn{fireable}(1) \eqdef \cn{pres}(\cn{DNAdam}) \otimes \cn{pres}(\cn{Mdm2}) 
        \otimes \cn{dont\_care}(\cn{p53}) \\
       s\_\cn{not\_fireable}(1) \eqdef \\ 
         ( (\cn{abs}(\cn{DNAdam}) \otimes \cn{pres}(\cn{Mdm2})) 
           ~\oplus~  (\cn{pres}(\cn{DNAdam}) \otimes \cn{abs}(\cn{Mdm2}))
           ~\oplus~  (\cn{abs}(\cn{DNAdam}) \otimes \cn{abs}(\cn{Mdm2})) )
         ~\otimes~ \cn{dont\_care}(\cn{p53})$ }
%
$$\begin{array}{c}
\dfrac
   {\dfrac 
      {P_{11F} \qquad P_{11N}}
      {(\cn{pres}(\cn{p53}) ~\otimes~ \cn{pres}(\cn{Mdm2}))  ~@~ n,
        (\cn{pres}(\cn{DNAdam}) ~\oplus~ \cn{abs}(\cn{DNAdam})) ~@~ n
        \vdashseq 
        \cdots \oplus \cdots ~@~ n} ~ \oplus L
   }
   {(\cn{pres}(\cn{p53}) ~\otimes~ \cn{pres}(\cn{Mdm2}))  ~@~ n 
    \vdashseq ~(s\_\cn{fireable}(1) \mathbin{\&} \delay{1} \mathcal{R}) \oplus s\_\cn{not\_fireable}(1) ~@~ n}
\end{array}$$

$P_{11F}$ :
$$\begin{array}{c}
\hspace{-0.5cm}
\dfrac
      {\dfrac 
          {P_{11Ff} \qquad P_{11Fr}}
          {\cn{pres}(\cn{p53}) ~\otimes~ \cn{pres}(\cn{Mdm2}) 
           ~\otimes~ \cn{pres}(\cn{DNAdam}) ~@~ n
          \vdashseq s\_\cn{fireable}(1) \mathbin{\&} \delay{1} \mathcal{R} ~@~ n}
          ~ \mathbin{\&} R
      }
      {(\cn{pres}(\cn{p53}) ~\otimes~ \cn{pres}(\cn{Mdm2}))  ~@~ n,  
       \cn{pres}(\cn{DNAdam}) ~@~ n \vdashseq 
       (s\_\cn{fireable}(1) \mathbin{\&} \delay{1} \mathcal{R}) \oplus s\_\cn{not\_fireable}(1) ~@~ n
      } ~ \oplus R_1
\end{array}$$
$P_{11Ff}$ :
$$\begin{array}{c}
\hspace{-0.5cm}
\cn{pres}(\cn{p53}) ~\otimes~ \cn{pres}(\cn{Mdm2}) ~\otimes~ \cn{pres}(\cn{DNAdam}) ~@~ n 
   \vdashseq s\_\cn{fireable}(1) ~@~ n ~~ [...; hyp]
\end{array}$$
$P_{11Fr}$ :
{\footnotesize
$$\begin{array}{c}
\hspace{-1.0cm}
\dfrac
  {\begin{array}{l} \cn{pres}(\cn{DNAdam}) \otimes \cn{pres}(\cn{Mdm2}) ~@~ n \\
       \vdashseq \cn{pres}(\cn{DNAdam}) \otimes \cn{pres}(\cn{Mdm2}) ~@~ n 
   \end{array}
   \quad 
   \dfrac
     {\dfrac
         {\cdots  
             \vdashseq \cn{pres}(\cn{p53}) \otimes \cn{abs}(\cn{Mdm2}) ~@~ n+1}
        {\cdots  
             \vdashseq ( \cn{pres}(\cn{p53}) \otimes \cn{abs}(\cn{Mdm2}) ) \oplus \cdots ~@~ n+1}   \oplus R
     \quad 
      \dfrac
         {\cdots  
         \vdashseq \cn{pres}(\cn{DNAdam}) ~@~ n+1} 
         {\cdots  
         \vdashseq \cn{dont\_care}(\cn{DNAdam}) ~@~ n+1}   \oplus R
    }
     {\cn{pres}(\cn{p53}) ~@~ n, ~
       (\cn{pres}(\cn{DNAdam}) ~\otimes~ \cn{abs}(\cn{Mdm2})) ~@~ n+1,~ 
       \cn{unchanged}(\cn{p53} ) ~@~ n 
       \vdashseq \mathcal{R} ~@~ n+1 
    } 
  } 
  {s\_\cn{rule}(1) ~@~ n, 
   ~(\cn{pres}(\cn{p53}) ~\otimes~ \cn{pres}(\cn{Mdm2}) ~\otimes~ \cn{pres}(\cn{DNAdam})) ~@~ n 
   \vdashseq \delay{1} \mathcal{R} ~@~ n
  } \limp L
\end{array}$$
}
$P_{11N}$ :
$$\begin{array}{c}
\hspace{-0.5cm}
    \dfrac
      {\cn{pres}(\cn{p53}) ~@~ n,
       (\cn{abs}(\cn{DNAdam}) ~\otimes~ \cn{pres}(\cn{Mdm2})) ~@~ n
       \vdashseq s\_\cn{not\_fireable}(1) ~@~ n ~~ [...; hyp]}
      {(\cn{pres}(\cn{p53}) ~\otimes~ \cn{pres}(\cn{Mdm2}))  ~@~ n,  
       \cn{abs}(\cn{DNAdam}) ~@~ n \vdashseq 
       (s\_\cn{fireable}(1) \mathbin{\&} \delay{1} \mathcal{R}) \oplus s\_\cn{not\_fireable}(1) ~@~ n
      } ~ \oplus R_2
\end{array}$$

$P_{12}$ :
\ednote[]{$s\_\cn{rule}(2) \eqdef 
   \cn{abs}(\cn{Mdm2}) \otimes \cn{abs}(\cn{p53}) 
   \limp \delay{1} (\cn{abs}(\cn{Mdm2}) \otimes \cn{pres}(\cn{p53})) \otimes \downarrow u.~ \cn{unchanged}(\cn{DNAdam},u))$}
\ednote[]{$s\_\cn{fireable}(2) \eqdef \cn{abs}(\cn{Mdm2})  \otimes \cn{abs}(\cn{p53}) 
        \otimes \cn{dont\_care}(\cn{DNAdam}) \\
     s\_\cn{not\_fireable}(2) \eqdef \\
        ( (\cn{pres}(\cn{Mdm2})  \otimes \cn{abs}(\cn{p53})) 
          ~\oplus~ (\cn{abs}(\cn{Mdm2})  \otimes \cn{pres}(\cn{p53}))
          ~\oplus~ (\cn{pres}(\cn{Mdm2})  \otimes \cn{pres}(\cn{p53})) 
        ~\otimes~ \cn{dont\_care}(\cn{DNAdam})$ }
$$\begin{array}{c}
\dfrac
      {\cn{pres}(\cn{p53}) ~\otimes~ \cn{pres}(\cn{Mdm2})  ~@~ n 
    \vdashseq s\_\cn{not\_fireable}(2) ~@~ n ~~[hyp]}
   {\cn{pres}(\cn{p53}) ~\otimes~ \cn{pres}(\cn{Mdm2})  ~@~ n 
    \vdashseq ~(s\_\cn{fireable}(2) \mathbin{\&} \delay{1} \mathcal{R}) \oplus s\_\cn{not\_fireable}(2) ~@~ n} ~\oplus R_2
\end{array}$$

$P_{13}$ :
\ednote[]{$s\_\cn{rule}(3) \eqdef 
\cn{pres}(\cn{p53}) \otimes \cn{abs}(\cn{Mdm2}) 
            \limp \delay{1} (\cn{pres}(\cn{p53}) \otimes \cn{pres}(\cn{Mdm2})) \otimes \downarrow u.~ \cn{unchanged}(\cn{DNAdam},u))$}
\ednote[]{$s\_\cn{fireable}(3) \eqdef \cn{pres}(\cn{p53}) \otimes \cn{abs}(\cn{Mdm2}) 
        \otimes \cn{dont\_care}(\cn{DNAdam}) \\
     s\_\cn{not\_fireable}(3) \eqdef \\
        ( (\cn{abs}(\cn{p53}) \otimes \cn{abs}(\cn{Mdm2}))
          ~\oplus~ (\cn{pres}(\cn{p53}) \otimes \cn{pres}(\cn{Mdm2}))
          ~\oplus~ (\cn{abs}(\cn{p53}) \otimes \cn{pres}(\cn{Mdm2})) )
        ~\otimes~ \cn{dont\_care}(\cn{DNAdam}) $}
$$\begin{array}{c}
\dfrac
      {\cn{pres}(\cn{p53}) ~\otimes~ \cn{pres}(\cn{Mdm2})  ~@~ n 
    \vdashseq s\_\cn{not\_fireable}(3) ~@~ n ~~[hyp]}
   {\cn{pres}(\cn{p53}) ~\otimes~ \cn{pres}(\cn{Mdm2})  ~@~ n 
    \vdashseq ~(s\_\cn{fireable}(3) \mathbin{\&} \delay{1} \mathcal{R}) \oplus s\_\cn{not\_fireable}(3) ~@~ n} ~\oplus R_2
\end{array}$$

$P_{14}$ :
\ednote[]{$s\_\cn{rule}(4) \eqdef 
    \cn{pres}(\cn{Mdm2}) \otimes \cn{pres}(\cn{p53}) 
    \limp \delay{1} (\cn{pres}(\cn{Mdm2}) \otimes \cn{abs}(\cn{p53})) ~\otimes~ \downarrow u.~ \cn{unchanged}(\cn{DNAdam},u))$}
\ednote[]{$s\_\cn{fireable}(4) \eqdef \cn{pres}(\cn{Mdm2}) \otimes \cn{pres}(\cn{p53}) 
        \otimes \cn{dont\_care}(\cn{DNAdam}) \\
    s\_\cn{not\_fireable}(4) \eqdef \\
        ( (\cn{abs}(\cn{Mdm2}) \otimes \cn{pres}(\cn{p53})) 
          ~\oplus~ (\cn{pres}(\cn{Mdm2}) \otimes \cn{abs}(\cn{p53})) 
          ~\oplus~ (\cn{abs}(\cn{Mdm2}) \otimes \cn{abs}(\cn{p53})) )
        ~\otimes~ \cn{dont\_care}(\cn{DNAdam})$ }
$$\begin{array}{c}
\hspace{-0.5cm}
\dfrac
   {\dfrac
      {P_{14f} \qquad P_{14r}}
      {(\cn{pres}(\cn{p53}) ~\otimes~ \cn{pres}(\cn{Mdm2})) ~@~ n 
       \vdashseq s\_\cn{fireable}(4) \mathbin{\&} \delay{1} \mathcal{R} ~@~ n
      } \mathbin{\&} R
   }
   {(\cn{pres}(\cn{p53}) ~\otimes~ \cn{pres}(\cn{Mdm2}))  ~@~ n 
    \vdashseq ~(s\_\cn{fireable}(4) \mathbin{\&} \delay{1} \mathcal{R}) \oplus s\_\cn{not\_fireable}(4) ~@~ n} \oplus R_1
\end{array}$$
$P_{14f}$ :
$$\begin{array}{c}
\hspace{-0.5cm}
\dfrac
   {\cdots \vdashseq \cn{pres}(\cn{p53}) ~\otimes~ \cn{pres}(\cn{Mdm2}) ~@~ n ~~[hyp]
    \qquad 
    \cn{well\_defined_1}(\cn{DNAdam}) ~@~ n \vdashseq \cn{dont\_care}(\cn{DNAdam}) ~@~ n}
   {\cn{pres}(\cn{p53}) ~\otimes~ \cn{pres}(\cn{Mdm2}) ~@~ n \vdashseq s\_\cn{fireable}(4) ~@~ n}       
\end{array}$$
$P_{14r}$ :
{\small
$$\begin{array}{c}
\hspace{-1.1cm}
\dfrac
  {\begin{array}{l} \cn{pres}(\cn{p53}) \otimes \cn{pres}(\cn{Mdm2}) ~@~ n \\
       \vdashseq \cn{pres}(\cn{p53}) \otimes \cn{pres}(\cn{Mdm2}) ~@~ n 
   \end{array}
  ~~ 
   \dfrac
     {\dfrac
         {\cdots  
             \vdashseq \cn{abs}(\cn{p53}) \otimes \cn{pres}(\cn{Mdm2}) ~@~ n+1}
        {\cdots  
             \vdashseq ( \cn{abs}(\cn{p53}) \otimes \cn{pres}(\cn{Mdm2}) ) \oplus \cdots ~@~ n+1}  
     \quad 
      \dfrac
         { \cn{well\_defined_1}(\cn{DNAdam}) ~@~ n+1
         \vdashseq \cdots } 
         {\vdashseq \cn{dont\_care}(\cn{DNAdam}) ~@~ n+1} 
    }
     {\cn{pres}(\cn{Mdm2}) ~\otimes~ \cn{abs}(\cn{p53}) ~@~ n+1,~ 
       \cn{unchanged}(\cn{DNAdam} ) ~@~ n 
       \vdashseq \mathcal{R} ~@~ n+1 
    } 
  } 
  {s\_\cn{rule}(4) ~@~ n, 
   ~(\cn{pres}(\cn{p53}) ~\otimes~ \cn{pres}(\cn{Mdm2})) ~@~ n 
   \vdashseq \delay{1} \mathcal{R} ~@~ n
  } \limp L
\end{array}$$
}

$P_{15}$ :
\ednote[]{$s\_\cn{rule}(5) \eqdef 
   \cn{pres}(\cn{p53}) \otimes \cn{pres}(\cn{DNAdam}) 
   \limp \delay{1} (\cn{abs}(\cn{p53}) \otimes \cn{abs}(\cn{DNAdam})) 
      \otimes \downarrow u.~ \cn{unchanged}(\cn{Mdm2},u))$}
\ednote[]{$s\_\cn{fireable}(5) \eqdef \cn{pres}(\cn{p53}) \otimes \cn{pres}(\cn{DNAdam}) 
        \otimes \cn{dont\_care}(\cn{Mdm2}) \\
    s\_\cn{not\_fireable}(5) \eqdef \\
        ( (\cn{abs}(\cn{p53}) \otimes \cn{pres}(\cn{DNAdam}))
          ~\oplus~ (\cn{pres}(\cn{p53}) \otimes \cn{abs}(\cn{DNAdam}))
          ~\oplus~ (\cn{abs}(\cn{p53}) \otimes \cn{abs}(\cn{DNAdam})) )
        ~\otimes~ \cn{dont\_care}(\cn{Mdm2})$ }
$$\begin{array}{c}
\dfrac
   {\dfrac 
      {P_{15F} \qquad P_{15N}}
      {(\cn{pres}(\cn{p53}) ~\otimes~ \cn{pres}(\cn{Mdm2}))  ~@~ n,
        (\cn{pres}(\cn{DNAdam}) ~\oplus~ \cn{abs}(\cn{DNAdam})) ~@~ n
        \vdashseq 
        \cdots \oplus \cdots ~@~ n} ~ \oplus L
   }
   {(\cn{pres}(\cn{p53}) ~\otimes~ \cn{pres}(\cn{Mdm2}))  ~@~ n 
    \vdashseq ~(s\_\cn{fireable}(5) \mathbin{\&} \delay{1} \mathcal{R}) \oplus s\_\cn{not\_fireable}(5) ~@~ n}
\end{array}$$
$P_{15F}$ :
$$\begin{array}{c}
\hspace{-0.5cm}
\dfrac
      {\dfrac 
          {P_{15Ff} \qquad P_{15Fr}}
          {\cn{pres}(\cn{p53}) ~\otimes~ \cn{pres}(\cn{Mdm2}) 
           ~\otimes~ \cn{pres}(\cn{DNAdam}) ~@~ n
          \vdashseq s\_\cn{fireable}(5) \mathbin{\&} \delay{1} \mathcal{R} ~@~ n}
          ~ \mathbin{\&} R
      }
      {(\cn{pres}(\cn{p53}) ~\otimes~ \cn{pres}(\cn{Mdm2}))  ~@~ n,  
       \cn{pres}(\cn{DNAdam}) ~@~ n \vdashseq 
       (s\_\cn{fireable}(5) \mathbin{\&} \delay{1} \mathcal{R}) \oplus s\_\cn{not\_fireable}(5) ~@~ n
      } ~ \oplus R_1
\end{array}$$
$P_{15Ff}$ :
$$\begin{array}{c}
\hspace{-0.5cm}
\cn{pres}(\cn{p53}) ~\otimes~ \cn{pres}(\cn{Mdm2}) ~\otimes~ \cn{pres}(\cn{DNAdam}) ~@~ n 
   \vdashseq s\_\cn{fireable}(5) ~@~ n ~~ [...; hyp]
\end{array}$$
$P_{15Fr}$ :
{\footnotesize
$$\begin{array}{c}
\hspace{-1.0cm}
\dfrac
  {\begin{array}{l} \cn{pres}(\cn{p53}) \otimes \cn{pres}(\cn{DNAdam}) ~@~ n \\
       \vdashseq \cn{pres}(\cn{p53}) \otimes \cn{pres}(\cn{DNAdam}) ~@~ n 
   \end{array}
   \quad 
   \dfrac
     {\dfrac
         {\cdots 
             \vdashseq \cn{abs}(\cn{p53}) \otimes \cn{pres}(\cn{Mdm2}) ~@~ n+1}
        {\cdots 
             \vdashseq ( \cn{abs}(\cn{p53}) \otimes \cn{pres}(\cn{Mdm2}) ) \oplus \cdots ~@~ n+1}   \oplus R
     \quad 
      \dfrac
         {\cdots 
         \vdashseq \cn{pres}(\cn{DNAdam}) ~@~ n+1} 
         {\cdots 
         \vdashseq \cn{dont\_care}(\cn{DNAdam}) ~@~ n+1}   \oplus R
    }
     {\cn{pres}(\cn{Mdm2}) ~@~ n, ~
       (\cn{abs}(\cn{p53}) ~\otimes~ \cn{abs}(\cn{DNAdam})) ~@~ n+1,~ 
       \cn{unchanged}(\cn{Mdm2} ) ~@~ n 
       \vdashseq \mathcal{R} ~@~ n+1 
    } 
  } 
  {s\_\cn{rule}(5) ~@~ n, 
   ~(\cn{pres}(\cn{p53}) ~\otimes~ \cn{pres}(\cn{Mdm2}) ~\otimes~ \cn{pres}(\cn{DNAdam})) ~@~ n 
   \vdashseq \delay{1} \mathcal{R} ~@~ n
  } \limp L
\end{array}$$
}
$P_{15N}$ :
$$\begin{array}{c}
\hspace{-0.5cm}
    \dfrac
      {\cn{pres}(\cn{p53}) ~\otimes~ \cn{abs}(\cn{DNAdam}) ~@~ n,~
       \cn{pres}(\cn{Mdm2}) ~@~ n
       \vdashseq s\_\cn{not\_fireable}(5) ~@~ n ~~ [...; hyp]}
      {(\cn{pres}(\cn{p53}) ~\otimes~ \cn{pres}(\cn{Mdm2}))  ~@~ n,  
       \cn{abs}(\cn{DNAdam}) ~@~ n \vdashseq 
       (s\_\cn{fireable}(5) \mathbin{\&} \delay{1} \mathcal{R}) \oplus s\_\cn{not\_fireable}(5) ~@~ n
      } ~ \oplus R_2
\end{array}$$
$P_{16}$ :
\ednote[]{$s\_\cn{rule}(6) \eqdef 
\cn{abs}(\cn{DNAdam}) \otimes \cn{abs}(\cn{Mdm2}) 
            \limp \delay{1} (\cn{abs}(\cn{DNAdam}) \otimes \cn{pres}(\cn{Mdm2})) \otimes \downarrow u.~ \cn{unchanged}(\cn{p53},u))$}
\ednote[]{$s\_\cn{fireable}(6) \eqdef \cn{abs}(\cn{DNAdam}) \otimes \cn{abs}(\cn{Mdm2}) 
        \otimes \cn{dont\_care}(\cn{p53}) \\
   s\_\cn{not\_fireable}(6) \eqdef \\
        ( (\cn{pres}(\cn{DNAdam}) \otimes \cn{abs}(\cn{Mdm2}))
          ~\oplus~  (\cn{abs}(\cn{DNAdam}) \otimes \cn{pres}(\cn{Mdm2}))
          ~\oplus~  (\cn{pres}(\cn{DNAdam}) \otimes \cn{pres}(\cn{Mdm2})) )
        ~\otimes~ \cn{dont\_care}(\cn{p53})$}
$$\begin{array}{c}
\dfrac
      {\cn{pres}(\cn{p53}) ~\otimes~ \cn{pres}(\cn{Mdm2})  ~@~ n 
    \vdashseq s\_\cn{not\_fireable}(6) ~@~ n ~~[hyp]}
   {\cn{pres}(\cn{p53}) ~\otimes~ \cn{pres}(\cn{Mdm2})  ~@~ n 
    \vdashseq ~(s\_\cn{fireable}(6) \mathbin{\&} \delay{1} \mathcal{R}) \oplus s\_\cn{not\_fireable}(6) ~@~ n} ~\oplus R_2
\end{array}$$

and $P_2$ is
$$\begin{array}{c}
\dfrac
   {P_{21} \quad P_{22} \quad P_{23} \quad P_{24} \quad P_{25} \quad P_{26}}
   {(\cn{abs}(\cn{p53}) ~\otimes~ \cn{abs}(\cn{Mdm2}))  ~@~ n 
    \vdashseq \forall r: [1..6]. ~ (s\_\cn{fireable}(r) \mathbin{\&} \delay{1} \mathcal{R}) \oplus s\_\cn{not\_fireable}(r) ~@~ n} ~ case_r
\end{array}$$

with $P_{21}$ :
\ednote[]{$s\_\cn{rule}(1) \eqdef \cn{pres}(\cn{DNAdam}) \otimes \cn{pres}(\cn{Mdm2}) 
 \limp \delay{1} (\cn{pres}(\cn{DNAdam}) \otimes \cn{abs}(\cn{Mdm2})) ~\otimes \downarrow u.~ \cn{unchanged}(\cn{p53},u)$}
\ednote[]{
     $s\_\cn{fireable}(1) \eqdef \cn{pres}(\cn{DNAdam}) \otimes \cn{pres}(\cn{Mdm2}) 
        \otimes \cn{dont\_care}(\cn{p53}) \\
       s\_\cn{not\_fireable}(1) \eqdef \\ 
         ( (\cn{abs}(\cn{DNAdam}) \otimes \cn{pres}(\cn{Mdm2})) 
           ~\oplus~  (\cn{pres}(\cn{DNAdam}) \otimes \cn{abs}(\cn{Mdm2}))
           ~\oplus~  (\cn{abs}(\cn{DNAdam}) \otimes \cn{abs}(\cn{Mdm2})) )
         ~\otimes~ \cn{dont\_care}(\cn{p53})$ }
$$\begin{array}{c}
\dfrac
   {\cn{abs}(\cn{p53}) ~\otimes~ \cn{abs}(\cn{Mdm2})  ~@~ n 
    \vdashseq s\_\cn{not\_fireable}(1) ~@~ n}  
   {\cn{abs}(\cn{p53}) ~\otimes~ \cn{abs}(\cn{Mdm2})  ~@~ n 
    \vdashseq ~(s\_\cn{fireable}(1) \mathbin{\&} \delay{1} \mathcal{R}) \oplus s\_\cn{not\_fireable}(1) ~@~ n} ~\oplus R_2
\end{array}$$

$P_{22}$ :
\ednote[]{$s\_\cn{rule}(2) \eqdef 
   \cn{abs}(\cn{Mdm2}) \otimes \cn{abs}(\cn{p53}) 
   \limp \delay{1} (\cn{abs}(\cn{Mdm2}) \otimes \cn{pres}(\cn{p53})) \otimes \downarrow u.~ \cn{unchanged}(\cn{DNAdam},u))$}
\ednote[]{$s\_\cn{fireable}(2) \eqdef \cn{abs}(\cn{Mdm2})  \otimes \cn{abs}(\cn{p53}) 
        \otimes \cn{dont\_care}(\cn{DNAdam}) \\
     s\_\cn{not\_fireable}(2) \eqdef \\
        ( (\cn{pres}(\cn{Mdm2})  \otimes \cn{abs}(\cn{p53})) 
          ~\oplus~ (\cn{abs}(\cn{Mdm2})  \otimes \cn{pres}(\cn{p53}))
          ~\oplus~ (\cn{pres}(\cn{Mdm2})  \otimes \cn{pres}(\cn{p53})) 
        ~\otimes~ \cn{dont\_care}(\cn{DNAdam})$ }
$$\begin{array}{c}
\hspace{-0.5cm}
\dfrac
   {\dfrac
      {P_{22f} \qquad P_{22r}}
      {(\cn{abs}(\cn{p53}) ~\otimes~ \cn{abs}(\cn{Mdm2})) ~@~ n 
       \vdashseq s\_\cn{fireable}(2) \mathbin{\&} \delay{1} \mathcal{R} ~@~ n
      } \mathbin{\&} R
   }
   {(\cn{abs}(\cn{p53}) ~\otimes~ \cn{abs}(\cn{Mdm2}))  ~@~ n 
    \vdashseq ~(s\_\cn{fireable}(2) \mathbin{\&} \delay{1} \mathcal{R}) \oplus s\_\cn{not\_fireable}(2) ~@~ n} \oplus R_1
\end{array}$$
$P_{22f}$ :
$$\begin{array}{c}
\hspace{-0.5cm}
\dfrac
   {\cdots \vdashseq \cn{abs}(\cn{p53}) ~\otimes~ \cn{abs}(\cn{Mdm2}) ~@~ n ~~[hyp]
    \qquad 
    \cn{well\_defined_1}(\cn{DNAdam}) ~@~ n \vdashseq \cn{dont\_care}(\cn{DNAdam}) ~@~ n}
   {\cn{abs}(\cn{p53}) ~\otimes~ \cn{abs}(\cn{Mdm2}) ~@~ n \vdashseq s\_\cn{fireable}(2) ~@~ n}       
\end{array}$$
$P_{22r}$ :
{\small
$$\begin{array}{c}
\hspace{-1.0cm}
\dfrac
  {\begin{array}{l} \cn{abs}(\cn{p53}) \otimes \cn{abs}(\cn{Mdm2}) ~@~ n \\
       \vdashseq \cn{abs}(\cn{p53}) \otimes \cn{abs}(\cn{Mdm2}) ~@~ n 
   \end{array}
  ~~ 
   \dfrac
     {\dfrac
         {\cdots  
             \vdashseq \cn{pres}(\cn{p53}) \otimes \cn{abs}(\cn{Mdm2}) ~@~ n+1}
        {\cdots  
             \vdashseq ( \cn{pres}(\cn{p53}) \otimes \cn{abs}(\cn{Mdm2}) ) \oplus \cdots ~@~ n+1}  
     \quad 
      \dfrac
         { \cn{well\_defined_1}(\cn{DNAdam}) ~@~ n+1
         \vdashseq \cdots } 
         {\vdashseq \cn{dont\_care}(\cn{DNAdam}) ~@~ n+1} 
    }
     {\cn{abs}(\cn{Mdm2}) ~\otimes~ \cn{pres}(\cn{p53}) ~@~ n+1,~ 
       \cn{unchanged}(\cn{DNAdam} ) ~@~ n 
       \vdashseq \mathcal{R} ~@~ n+1 
    } 
  } 
  {s\_\cn{rule}(2) ~@~ n, 
   ~(\cn{abs}(\cn{p53}) ~\otimes~ \cn{abs}(\cn{Mdm2})) ~@~ n 
   \vdashseq \delay{1} \mathcal{R} ~@~ n
  } \limp L
\end{array}$$
}

$P_{23}$ :
\ednote[]{$s\_\cn{rule}(3) \eqdef 
   \cn{pres}(\cn{p53}) \otimes \cn{abs}(\cn{Mdm2}) 
   \limp \delay{1} (\cn{pres}(\cn{p53}) \otimes \cn{pres}(\cn{Mdm2})) \otimes \downarrow u.~ \cn{unchanged}(\cn{DNAdam},u))$}
\ednote[]{$s\_\cn{fireable}(3) \eqdef \cn{pres}(\cn{p53}) \otimes \cn{abs}(\cn{Mdm2}) 
        \otimes \cn{dont\_care}(\cn{DNAdam}) \\
     s\_\cn{not\_fireable}(3) \eqdef \\
        ( (\cn{abs}(\cn{p53}) \otimes \cn{abs}(\cn{Mdm2}))
          ~\oplus~ (\cn{pres}(\cn{p53}) \otimes \cn{pres}(\cn{Mdm2}))
          ~\oplus~ (\cn{abs}(\cn{p53}) \otimes \cn{pres}(\cn{Mdm2})) )
        ~\otimes~ \cn{dont\_care}(\cn{DNAdam}) $}
$$\begin{array}{c}
\dfrac
   {\cn{abs}(\cn{p53}) ~\otimes~ \cn{abs}(\cn{Mdm2})  ~@~ n 
    \vdashseq s\_\cn{not\_fireable}(3) ~@~ n}  
   {\cn{abs}(\cn{p53}) ~\otimes~ \cn{abs}(\cn{Mdm2})  ~@~ n 
    \vdashseq ~(s\_\cn{fireable}(3) \mathbin{\&} \delay{1} \mathcal{R}) \oplus s\_\cn{not\_fireable}(3) ~@~ n} ~\oplus R_2
\end{array}$$

$P_{24}$ :
\ednote[]{$s\_\cn{rule}(4) \eqdef 
    \cn{pres}(\cn{Mdm2}) \otimes \cn{pres}(\cn{p53}) 
    \limp \delay{1} (\cn{pres}(\cn{Mdm2}) \otimes \cn{abs}(\cn{p53})) ~\otimes~ \downarrow u.~ \cn{unchanged}(\cn{DNAdam},u))$}
\ednote[]{$s\_\cn{fireable}(4) \eqdef \cn{pres}(\cn{Mdm2}) \otimes \cn{pres}(\cn{p53}) 
        \otimes \cn{dont\_care}(\cn{DNAdam}) \\
    s\_\cn{not\_fireable}(4) \eqdef \\
        ( (\cn{abs}(\cn{Mdm2}) \otimes \cn{pres}(\cn{p53})) 
          ~\oplus~ (\cn{pres}(\cn{Mdm2}) \otimes \cn{abs}(\cn{p53})) 
          ~\oplus~ (\cn{abs}(\cn{Mdm2}) \otimes \cn{abs}(\cn{p53})) )
        ~\otimes~ \cn{dont\_care}(\cn{DNAdam})$ }
$$\begin{array}{c}
\dfrac
   {\cn{abs}(\cn{p53}) ~\otimes~ \cn{abs}(\cn{Mdm2})  ~@~ n 
    \vdashseq s\_\cn{not\_fireable}(4) ~@~ n}  
   {\cn{abs}(\cn{p53}) ~\otimes~ \cn{abs}(\cn{Mdm2})  ~@~ n 
    \vdashseq ~(s\_\cn{fireable}(4) \mathbin{\&} \delay{1} \mathcal{R}) \oplus s\_\cn{not\_fireable}(4) ~@~ n} ~\oplus R_2
\end{array}$$

$P_{25}$ :
\ednote[]{$s\_\cn{rule}(5) \eqdef 
   \cn{pres}(\cn{p53}) \otimes \cn{pres}(\cn{DNAdam}) 
   \limp \delay{1} (\cn{abs}(\cn{p53}) \otimes \cn{abs}(\cn{DNAdam})) \otimes \downarrow u.~ \cn{unchanged}(\cn{Mdm2},u))$}
\ednote[]{$s\_\cn{fireable}(5) \eqdef \cn{pres}(\cn{p53}) \otimes \cn{pres}(\cn{DNAdam}) 
        \otimes \cn{dont\_care}(\cn{Mdm2}) \\
    s\_\cn{not\_fireable}(5) \eqdef \\
        ( (\cn{abs}(\cn{p53}) \otimes \cn{pres}(\cn{DNAdam}))
          ~\oplus~ (\cn{pres}(\cn{p53}) \otimes \cn{abs}(\cn{DNAdam}))
          ~\oplus~ (\cn{abs}(\cn{p53}) \otimes \cn{abs}(\cn{DNAdam})) )
        ~\otimes~ \cn{dont\_care}(\cn{Mdm2})$ }
$$\begin{array}{c}
\dfrac
   {\cn{abs}(\cn{p53}) ~\otimes~ \cn{abs}(\cn{Mdm2})  ~@~ n 
    \vdashseq s\_\cn{not\_fireable}(5) ~@~ n}  
   {\cn{abs}(\cn{p53}) ~\otimes~ \cn{abs}(\cn{Mdm2})  ~@~ n 
    \vdashseq ~(s\_\cn{fireable}(5) \mathbin{\&} \delay{1} \mathcal{R}) \oplus s\_\cn{not\_fireable}(5) ~@~ n} ~\oplus R_2
\end{array}$$

$P_{26}$ :
\ednote[]{$s\_\cn{rule}(6) \eqdef 
\cn{abs}(\cn{DNAdam}) \otimes \cn{abs}(\cn{Mdm2}) 
            \limp \delay{1} (\cn{abs}(\cn{DNAdam}) \otimes \cn{pres}(\cn{Mdm2})) \otimes \downarrow u.~ \cn{unchanged}(\cn{p53},u))$}
\ednote[]{$s\_\cn{fireable}(6) \eqdef \cn{abs}(\cn{DNAdam}) \otimes \cn{abs}(\cn{Mdm2}) 
        \otimes \cn{dont\_care}(\cn{p53}) \\
   s\_\cn{not\_fireable}(6) \eqdef \\
        ( (\cn{pres}(\cn{DNAdam}) \otimes \cn{abs}(\cn{Mdm2}))
          ~\oplus~  (\cn{abs}(\cn{DNAdam}) \otimes \cn{pres}(\cn{Mdm2}))
          ~\oplus~  (\cn{pres}(\cn{DNAdam}) \otimes \cn{pres}(\cn{Mdm2})) )
        ~\otimes~ \cn{dont\_care}(\cn{p53})$}
$$\begin{array}{c}
\dfrac
   {\dfrac 
      {P_{26F} \qquad P_{26N}}
      {(\cn{abs}(\cn{p53}) ~\otimes~ \cn{abs}(\cn{Mdm2}))  ~@~ n,
        (\cn{pres}(\cn{DNAdam}) ~\oplus~ \cn{abs}(\cn{DNAdam})) ~@~ n
        \vdashseq 
        \cdots \oplus \cdots ~@~ n} ~ \oplus L
   }
   {(\cn{abs}(\cn{p53}) ~\otimes~ \cn{abs}(\cn{Mdm2}))  ~@~ n 
    \vdashseq ~(s\_\cn{fireable}(6) \mathbin{\&} \delay{1} \mathcal{R}) \oplus s\_\cn{not\_fireable}(6) ~@~ n}
\end{array}$$

$P_{26F}$ :
$$\begin{array}{c}
\hspace{-0.5cm}
\dfrac
      {\dfrac 
          {P_{26Ff} \qquad P_{26Fr}}
          {\cn{abs}(\cn{p53}) ~\otimes~ \cn{abs}(\cn{Mdm2}) 
           ~\otimes~ \cn{abs}(\cn{DNAdam}) ~@~ n
          \vdashseq s\_\cn{fireable}(6) \mathbin{\&} \delay{1} \mathcal{R} ~@~ n}
          ~ \mathbin{\&} R
      }
      {(\cn{abs}(\cn{p53}) ~\otimes~ \cn{abs}(\cn{Mdm2}))  ~@~ n,  
       \cn{abs}(\cn{DNAdam}) ~@~ n \vdashseq 
       (s\_\cn{fireable}(6) \mathbin{\&} \delay{1} \mathcal{R}) \oplus s\_\cn{not\_fireable}(6) ~@~ n
      } ~ \oplus R_1
\end{array}$$
$P_{26Ff}$ :
$$\begin{array}{c}
\hspace{-0.5cm}
\cn{abs}(\cn{p53}) ~\otimes~ \cn{abs}(\cn{Mdm2}) ~\otimes~ \cn{abs}(\cn{DNAdam}) ~@~ n 
   \vdashseq s\_\cn{fireable}(6) ~@~ n ~~ [...; hyp]
\end{array}$$
$P_{26Fr}$ :
{\footnotesize
$$\begin{array}{c}
\hspace{-1.0cm}
\dfrac
  {\begin{array}{l} \cn{abs}(\cn{DNAdam}) \otimes \cn{abs}(\cn{Mdm2}) ~@~ n \\
       \vdashseq \cn{abs}(\cn{DNAdam}) \otimes \cn{abs}(\cn{Mdm2}) ~@~ n 
   \end{array}
   \quad 
   \dfrac
     {\dfrac
         {\cdots 
             \vdashseq \cn{abs}(\cn{p53}) \otimes \cn{pres}(\cn{Mdm2}) ~@~ n+1}
        {\cdots 
             \vdashseq ( \cn{abs}(\cn{p53}) \otimes \cn{pres}(\cn{Mdm2}) ) \oplus \cdots ~@~ n+1}   \oplus R
     \quad 
      \dfrac
         {\cdots
         \vdashseq \cn{abs}(\cn{DNAdam}) ~@~ n+1} 
         {\cdots 
         \vdashseq \cn{dont\_care}(\cn{DNAdam}) ~@~ n+1}   \oplus R
    }
     {\cn{abs}(\cn{p53}) ~@~ n, ~
       (\cn{abs}(\cn{DNAdam}) ~\otimes~ \cn{pres}(\cn{Mdm2})) ~@~ n+1,~ 
       \cn{unchanged}(\cn{p53} ) ~@~ n 
       \vdashseq \mathcal{R} ~@~ n+1 
    } 
  } 
  {s\_\cn{rule}(6) ~@~ n, 
   ~(\cn{abs}(\cn{p53}) ~\otimes~ \cn{abs}(\cn{Mdm2}) ~\otimes~ \cn{abs}(\cn{DNAdam})) ~@~ n 
   \vdashseq \delay{1} \mathcal{R} ~@~ n
  } \limp L
\end{array}$$
}
$P_{26N}$ :
$$\begin{array}{c}
\hspace{-0.5cm}
    \dfrac
      {\cn{pres}(\cn{DNAdam}) ~\otimes~ \cn{abs}(\cn{Mdm2}) ~@~ n,~
       \cn{abs}(\cn{p53}) ~@~ n
       \vdashseq s\_\cn{not\_fireable}(6) ~@~ n ~~ [...; hyp]}
      {(\cn{abs}(\cn{p53}) ~\otimes~ \cn{abs}(\cn{Mdm2}))  ~@~ n,  
       \cn{pres}(\cn{DNAdam}) ~@~ n \vdashseq 
       (s\_\cn{fireable}(6) \mathbin{\&} \delay{1} \mathcal{R}) \oplus s\_\cn{not\_fireable}(6) ~@~ n
      } ~ \oplus R_2
\end{array}$$
\end{proof}

%

\end{document}